\documentclass[a4paper,UKenglish,cleveref, autoref,thm-restate]{lipics-v2021}

\usepackage[noadjust]{cite}
\usepackage[most]{tcolorbox}

\usepackage{datetime}
\usepackage{booktabs}
\usepackage{enumitem}
\usepackage{pifont}
\usepackage{caption}
\usepackage{subcaption}
\usepackage[T1]{fontenc}
\usepackage[utf8]{inputenc}
\usepackage{pgfplots}
\usepackage{amsmath, amssymb}
\usepackage{amsthm}
\usepackage{color}
\usepackage{cases}
\usepackage{xparse}
\usepackage{xargs}
\usepackage{appendix}
\usepackage{algorithm}
\usepackage{algpseudocode}
\usepackage{verbatim, xspace}
\usepackage{tikz}
\usepackage{mathtools}
\usepackage{stmaryrd}
\usepackage{xcolor}
\usepackage{framed}
\usepackage[framemethod=TikZ]{mdframed}
\usepackage{hyperref}

\usetikzlibrary{decorations.pathreplacing,angles,quotes}
\usetikzlibrary{arrows}

\Crefname{assumption}{Assumption}{Assumptions}
\Crefname{fact}{Fact}{Facts}
\Crefname{enumi}{Property}{Properties}
\Crefname{observation}{Observation}{Observations}
\Crefname{claim}{Claim}{Claims}
\Crefname{equation}{Equation}{Equations}

\newcommand{\floor}[1]{\left\lfloor #1 \right\rfloor}
\newcommand{\eps}{\varepsilon}
\newcommand{\Oh}{\mathcal{O}}
\newcommand{\Os}{\Oh^{\star}}
\newcommand{\nat}{\mathbb{N}}
\newcommand{\N}{\mathbb{N}}
\newcommand{\Z}{\mathbb{Z}}
\newcommand{\Ff}{\mathcal{F}}
\newcommand{\Dd}{\mathcal{D}}
\newcommand{\Rr}{\mathcal{R}}
\newcommand{\Ww}{\mathcal{W}}
\newcommand{\Zz}{\mathcal{Z}}
\newcommand{\poly}{\mathrm{poly}}
\newcommand{\DP}{\mathtt{DP}} 
\newcommand{\cost}{\mathtt{cost}}
\newcommand{\ans}{\mathtt{ans}}

\newcommand{\BP}{\mathtt{BP}}

\newcommand{\OPT}{\mathtt{OPT}}

\renewcommand{\leq}{\leqslant}
\renewcommand{\geq}{\geqslant}
\renewcommand{\le}{\leqslant}
\renewcommand{\ge}{\geqslant}

\renewcommand{\subset}{\subseteq}

\newcommand{\Cmax}{C_{\max}}

\DeclareMathOperator*{\argmin}{arg\,min} 

\newcommand{\defproblem}[3]{\begin{tcolorbox}[
            enhanced,
            colback=white,
            colframe=black,
            boxrule = 0.5pt,
            coltitle=black,
            title=#1,
            rounded corners,
            attach boxed title to top left={yshift=-10pt, xshift=6pt},
            bottom=-0.5mm,
            boxed title style={
                    interior style={fill=white},
                    frame hidden
                }
        ]
        \begin{tabularx}{12.5cm}{ r X p {0.5cm}}
            {}Input:     & #2        \\
            {}Task:  & #3
        \end{tabularx}
    \end{tcolorbox}
}

\newcommand{\binpacking}{{\textup{\textsc{Bin Packing}}}\xspace}
\newcommand{\threepartition}{{\textup{\textsc{$3$-way Partitioning}}}\xspace}
\newcommand{\PSwjCj}{{\textup{$P \mid \mid  \Sigma w_j C_j$}}\xspace}
\newcommand{\PmSwjCj}[1]{{\textup{$P {#1} \mid \mid  \Sigma w_j C_j$}}\xspace}
\newcommand{\PSwjUj}{{\textup{$P \mid \mid \Sigma w_j U_j$}}\xspace}
\newcommand{\PSCj}{{\textup{$P \mid \mid  \Sigma  C_j$}}\xspace}
\newcommand{\PSUj}{{\textup{$P \mid \mid \Sigma U_j$}}\xspace}
\newcommand{\PmSUj}[1]{{\textup{$P {#1} \mid \mid \Sigma U_j$}}\xspace}
\newcommand{\PSwjTj}{{\textup{$P \mid \mid  \Sigma w_j T_j$}}\xspace}
\newcommand{\PSTj}{{\textup{$P \mid \mid \Sigma T_j$}}\xspace}
\newcommand{\PSwjfj}{{\textup{$P \mid \mid \Sigma w_j f_j$}}\xspace}
\newcommand{\PCmax}{{\textup{$P \mid \mid \Cmax$}}\xspace}

\newcommand{\dc}{{\downarrow}}
\newcommand{\uc}{{\uparrow}}

\newlist{caseenum}{enumerate}{1}
\setlist[caseenum,1]{label= \textbf{\textsc{Case}} \Alph*.,
		ref= \Alph*,
		leftmargin=*, labelindent=0pt, itemsep=0.5em,topsep=0.25em}
\Crefname{caseenum}{Case}{Cases}
 
\hideLIPIcs
\author{Anubhav Dhar}{Max Planck Institute for Informatics, Saarbrücken, Germany \and Saarland University, Saarbrücken, Germany}{adhar@mpi-sws.org}{https://orcid.org/0009-0006-5922-8300}{Supported by the Deutsche Forschungsgemeinschaft (DFG, German Research Foundation) grant number 559177164.}
\author{Anita D\"urr}{ETH, Zurich, Switzerland}{anita.duerr@inf.ethz.ch}{https://orcid.org/0000-0003-0440-5008}{Part of this work was done while affiliated to Saarland University and Max Planck Institute for Informatics, Saarbrücken, Germany, where this work was part of the project TIPEA that has received funding from the European Research Council (ERC) under the European Union's Horizon 2020 research and innovation programme (grant agreement No. 850979).}
\author{Ahmed Ghazy}{CISPA Helmholtz Center for Information Security, Saarbrücken, Germany \and Saarland University, Saarbrücken, Germany}{ahmed.ghazy@cispa.de}{https://orcid.org/0009-0009-7414-5871}{}
\author{Jakob Greilhuber}{CISPA Helmholtz Center for Information Security, Saarbrücken, Germany \and Saarland University, Saarbrücken, Germany}{jakob.greilhuber@cispa.de}{https://orcid.org/0009-0001-8796-6400}{}
\author{Karol W\k{e}grzycki}{Max Planck Institute for Informatics, Saarbrücken, Germany}{kwegrzyc@mpi-inf.mpg.de}{https://orcid.org/0000-0001-9746-5733}{Supported by the Deutsche Forschungsgemeinschaft (DFG, German Research Foundation) grant number 559177164.}

\title{Faster Exponential Algorithms for Multi-Machine Scheduling Problems}

\nolinenumbers
 
\begin{document}

\authorrunning{A. Dhar, A. D\"urr, A. Ghazy, J. Greilhuber, K. W\k{e}grzycki}
\Copyright{Anubhav Dhar, Anita D\"urr, Ahmed Ghazy, Jakob Greilhuber, Karol W\k{e}grzycki}


\maketitle

\begin{abstract}
	Minimizing the weighted completion times ($P \mid \mid  \Sigma w_j C_j$) and weighted number of tardy jobs ($P \mid \mid  \Sigma w_j U_j$) on multiple identical machines are two classical NP-hard scheduling problems.
As shown by Lent\'e et al. (2014), both problems can be solved in time $\mathcal{O}^{\star}(3^n)$.\footnote{A function $f(n)$ is in $\Os(g(n))$ if it is in $\Oh(g(n) \cdot n^c)$ for some constant $c$.}
In this paper, we improve these bounds to $\mathcal{O}(2.755^n)$ and $\mathcal{O}^{\star}(2^n)$, respectively.
Our algorithm for $P \mid \mid  \Sigma w_j C_j$ exploits the meet-in-the-middle paradigm and an efficient data structure answering linear programming queries.
Additionally, when the number of machines is at most $6$, we show that the running time for $P \mid \mid  \Sigma w_j C_j$ can further be improved.

Both scheduling problems are generalizations of the classical \textsc{Bin Packing} problem, which can be solved in $\mathcal{O}^{\star}(2^n)$ time.
Improving this running time is an important open question. We show that, when assuming the Asymptotic Rank Conjecture (ARC), \textsc{Bin Packing} can be solved in time $\mathcal{O}((2-\varepsilon)^n)$ for some $\varepsilon >0$. Our algorithm makes use of two main ingredients: the recent $\mathcal{O}((2-\varepsilon)^n)$-time algorithm of Nederlof et al.~[SICOMP'23] for \textsc{Bin Packing} when the number of bins is a fixed constant, and the $\mathcal{O}((2-\varepsilon)^n)$-time algorithm of Björklund et al.~[SODA'25] for special instances of the \textsc{$3$-way Partitioning} problem when assuming ARC. \end{abstract}

\newpage

\section{Introduction}\label{sec:intro}

A central goal in the study of NP-hard problems is to determine the smallest possible running time, specifically for problems that admit exponential-time algorithms.
For many fundamental scheduling problems a $\Os(3^n)$-time algorithm constitutes a natural baseline, where $n$ is the number of jobs~\cite{lente14}.
Contrasting this, the Exponential Time Hypothesis (ETH) excludes $2^{o(n)}$-time algorithms~\cite{DBLP:journals/siamdm/JansenLL16}.
Therefore, it is a key question to find $\Os(C^n)$-time algorithms, for a constant $C > 0$ as small as possible.
In this paper, we address this question for two fundamental scheduling problems. 

Typically, 
the input of scheduling problems consists of $n$ jobs with processing times $p_1, \dots, p_n$, weights $w_1, \dots, w_n$ and due dates $d_1, \dots, d_n$, as well as a number $m \in \N$ of machines available. 
A schedule is an assignment of each job to a machine and a starting time, such that no two jobs are processed at the same time by the same machine. The goal is to find a schedule that minimizes some {cost}, where the used cost function varies from one scheduling problem to another. In this paper, we focus on two cost functions: the total weighted completion time $\sum_{j=1}^n w_j C_j$ and the weighted number of tardy jobs $\sum_{j=1}^n w_j U_j$, where $C_j$ denotes the completion time of job $j$ and $U_j$ is the indicator variable taking value 1 if job $j$ completes after its due date $d_j$ (i.e.~$C_j > d_j$, in which case job $j$ is called \emph{tardy}) and 0 otherwise. 
In the Graham et al.~\cite{graham1979optimization} notation\footnote{This useful notation is discussed in more detail later in the introduction.}
these problems are denoted by \PSwjCj and \PSwjUj, respectively.

The problems \PSwjCj and \PSwjUj are among the most fundamental problems in theoretical computer science.
Minimizing the weighted number of tardy jobs is NP-hard even on a single machine~\cite{lenstra1977complexity,Karp:72:Reducibility-among}. 
The problem \PSwjCj is polynomial-time solvable on a single machine~\cite{smith1956various}, but NP-hard already when $m = 2$~\cite{DBLP:journals/cacm/BrunoCS74,lenstra1977complexity}.
As a result, numerous approximation algorithms~\cite{DBLP:journals/jacm/Sahni76,DBLP:journals/dam/GensL81} as well as exact algorithms with parameterized or pseudopolynomial running time for (special cases of) both \PSwjUj 
\cite{lawler1969functional,moore1968n,DBLP:books/el/93/LawlerLKS93,DBLP:journals/corr/tightSETH,DBLP:conf/soda/Klein0R23,DBLP:journals/informs/HermelinMS24,DBLP:journals/jcss/AbboudBHS22,DBLP:journals/algorithmica/BringmannFHSW22,DBLP:journals/theoretics/FischerW25,heeger2024minimizing} and \PSwjCj
\cite{DBLP:books/el/93/LawlerLKS93,lawler1969functional,rothkopf1966scheduling,DBLP:journals/scheduling/KnopK18,DBLP:journals/corr/abs-2502-13631}
are present in the literature.

Lenté et al.~\cite{lente14} were the first to provide non-trivial exact exponential-time algorithms for \PSwjCj and \PSwjUj, concretely they proposed $\Os(3^n)$-time algorithms.
Contrasting this result, the work of Jansen et al.~\cite{DBLP:journals/siamdm/JansenLL16} implies that those problems do not admit any $2^{o(n)}$-time algorithms, unless ETH fails.
In this work, we improve upon the results of Lenté et al.~\cite{lente14}.

\begin{restatable}{theorem}{SumwjCj}\label{thm:SumwjCj}
    \PSwjCj can be solved in time $\Oh(2.755^n)$.
\end{restatable}

\begin{restatable}{theorem}{SumwjUj}\label{thm:SumwjUj}
        \PSwjUj can be solved in time $\Os(2^n)$.
    \end{restatable}

Both \cref{thm:SumwjCj,thm:SumwjUj} are proven using standard dynamic programming. 
Interestingly, for \PSwjCj{} the base cases (computing the dynamic programming tables for a constant number of machines) can be handled efficiently by exploiting the meet-in-the-middle paradigm combined with an efficient data structure that supports linear programming queries in $m$ dimensions.
For \PSwjUj we observe that Fast Subset Convolution~\cite{bjorklund2007fourier} can be used to solve the problem in the stated time, our concrete approach uses an algorithm of Moore~\cite{moore1968n} to handle the base case.

Additionally, we consider \PmSwjCj{} for a fixed number of machines $m$, and show that even faster algorithms are possible when $m \leq 6$.
For each integer $m \geq 1$, we denote the problem \PSwjCj where the number of machines is the fixed constant $m$ as \PmSwjCj{m}.

\begin{restatable}{theorem}{PmSumwjCj}\label{thm:PmSumwjCj}
    \PmSwjCj{m} can be solved in time:
    \begin{itemize}
        \item $\Os(m^{n/2})$ if $m \le 3$,
        \item $\Oh(2.389^n)$ if $m = 4$,
        \item $\Oh(2.726^n)$ if $m = 5$,
        \item $\Oh(2.733^n)$ if $m = 6$.
    \end{itemize}
\end{restatable}

\paragraph*{Bin Packing}
The scheduling problems we considered above are a generalization
of the scheduling problem where the cost function is given by the maximum completion time $\Cmax \coloneq \max_j C_j$, called the \emph{makespan} (see \cite{DBLP:journals/corr/tightSETH}). This scheduling problem is denoted as \PCmax in the Graham notation.
Its decision version coincides with the classical NP-hard \binpacking problem.

\binpacking can be solved in time $\Os(2^n)$~\cite{DBLP:journals/siamcomp/BjorklundHK09}, and faster algorithms are known when $m = \Oh(1)$ or $m = \Theta(n)$~\cite{DBLP:conf/esa/Nederlof16,nederlof2023faster}. However, in the general case, it remains a major open problem to improve the running time of $\Os(2^n)$~\cite{DBLP:journals/csr/Nederlof26}.
As a step towards answering that question, we can ask whether a faster algorithm is possible when assuming Strassen's \emph{Asymptotic Rank Conjecture} (ARC), a longstanding conjecture that provides significant algorithmic power.
Indeed, this has been a successful approach in the past. Faster algorithms for several problems were designed under the ARC, e.g. for $k$-\textsc{Set Cover}~\cite{DBLP:conf/stoc/BjorklundK24},~\textsc{Chromatic Number}~\cite{radu25}, \textsc{Hamiltonian Cycle} on bounded treewidth graphs~\cite{bjorklund_et_al:LIPIcs.ICALP.2026.36}, \textsc{Generalized Convolution}~\cite{DBLP:conf/sosa/BrandCLP26}, and $k$-\textsc{Orthogonal Vectors}~\cite{DBLP:conf/icalp/DurrKLW26}.
In this work, we follow the same strategy and show the existence of a faster than $2^n$-time algorithm for \binpacking when assuming ARC. 

\begin{restatable}{theorem}{BPARC}
\label{thm:bin-packing-arc}
    Assuming the Asymptotic Rank Conjecture, there exists a constant $\varepsilon>0$ such that the \binpacking problem on $n$ items can be solved with high probability in time $\Oh((2-\varepsilon)^n)$.
\end{restatable}

As it seems hard to obtain such a fast algorithm unconditionally, one might ask whether there exists a lower bound based on the \emph{Strong Exponential Time Hypothesis} (SETH) that rules out algorithms running in time $\Os((2-\varepsilon)^n)$ for any $\varepsilon > 0$.
However, such a lower bound would show that at most one of the SETH or the ARC is true, which would represent a major breakthrough in the understanding of these hypotheses.
Therefore, we view our result as a barrier to obtaining a tight SETH lower bound.
Contrasting this, since it is already known that at most one of the ARC and the \emph{Set Cover Conjecture (SCC)} are true~\cite{DBLP:conf/stoc/BjorklundK24}, it might be feasible to obtain a tight lower bound under SCC.
We leave a lower bound under SCC as an interesting direction for future work.

\paragraph*{Further scheduling problems}
In this paper we mainly focus on scheduling problems where the cost function is given by $\sum_{j=1}^n w_j C_j$ and $\sum_{j=1}^n w_j U_j$.
Since there are many other possibilities, one might wonder about the state-of-the-art for different scheduling problems, for example for unweighted variants or entirely different cost functions.
This leads to a systematic study of scheduling problems.

In order to discuss this further, we first provide additional background on problems on parallel machines.
The three-field notation $\alpha \mid \beta \mid \gamma$ introduced by Graham et al.~\cite{graham1979optimization} is a shorthand notation for scheduling problems. 
The first field $\alpha$ describes the machine environment, which in this paper will always be ``$P$'' or ``$Pm$'', denoting identical parallel machines where the number of machines is respectively part of the input or the fixed constant $m$.
The second field $\beta$ specifies additional constraints, and will always be empty in this paper. The third field $\gamma$ denotes the objective function that should be minimized. 
Natural objective functions are $\sum_{j=1}^n f_j$,
where $f_j$ is a cost function associated to the job $j$ which may take either of the following values:
\begin{itemize}
    \item the completion time $C_j$ of job $j$;
    \item the indicator variable $U_j$ which takes value 1 if job $j$ completes after its due date (i.e.~$C_j >d_j$) and 0 otherwise;
    \item the lateness $L_j \coloneq C_j - d_j$ of job $j$; or
    \item the tardiness $T_j \coloneq \max\{C_j - d_j, 0\}$ of job $j$;
\end{itemize}
as well as their weighted variants, i.e.~where the above measures are multiplied by $w_j$. 
By restricting the study to identical parallel machines without additional constraints (i.e.~the first field $\alpha$ is always ``$P$'', and the second field $\beta$ is always empty), this defines 8 distinct scheduling problems. 
Note, however, that the objectives $\sum_{j=1}^n w_j L_j$ and $\sum_{j=1}^n w_j C_j$ (as well as their unweighted variants) simply differ by the quantity $\sum_{j} d_j w_j$, which is independent of any underlying schedule.
Additionally, the problem of minimizing the total completion time \PSCj can be solved in polynomial time~\cite{DBLP:journals/cacm/BrunoCS74}.
This leaves 5 scheduling problems worth studying in the context of exponential-time algorithms. All of these remaining 5 problems were shown to be NP-hard early on: \PSwjCj and \PSUj are NP-hard already for $m=2$, and \PSwjUj is NP-hard already for $m=1$ \cite{lenstra1977complexity,DBLP:journals/cacm/BrunoCS74,Karp:72:Reducibility-among}.\footnote{For $m=1$, the problems \PSwjCj and \PSUj become polynomial-time solvable~\cite{moore1968n,maxwell1970sequencing,smith1956various}.}
Meanwhile, Du and Leung~\cite{DBLP:journals/mor/DuL90} showed that \PSTj (and thus \PSwjTj) is NP-hard when $m=1$.
On the other hand, Lenté et al.~\cite{lente14} proposed an algorithm solving any of the above scheduling problems in time $\Os(3^n)$.\footnote{In fact, the work of Lenté et al.~\cite{lente14} is applicable for an even wider range of cost functions $f_j$.} 

\cref{thm:SumwjUj,thm:SumwjCj} show that this upper bound can be improved for \PSwjCj and \PSwjUj.
We leave whether the upper bound can also be improved for \PSwjTj and \PSTj as an interesting question for future work.

\paragraph*{Organization of the paper}
We give a technical overview of our results in \cref{sec:tech_overview}.
After defining the used notation and stating important technical tools in \cref{sec:preliminaries}, we prove \cref{thm:SumwjCj,thm:PmSumwjCj}, our results about the weighted sum of completion times in \cref{sec:wjCj}.
\cref{thm:SumwjUj}, the result about the weighted number of tardy jobs, is proven in \cref{sec:wjUj}.
We conclude the paper with the proof of \cref{thm:arc_3partition} in \cref{sec:binpacking}. \section{Technical Overview}
\label{sec:tech_overview}
We now define the studied problems and give a high-level overview of \cref{thm:SumwjCj,thm:SumwjUj,thm:PmSumwjCj,thm:bin-packing-arc}.

\subsection{Weighted completion time}

\PSwjCj denotes the scheduling problem where the objective is to minimize the total weighted completion time.

\defproblem{\PSwjCj}{$n$ jobs, where the $j$-th job consists of processing time $p_j \in \mathbb{N}$ and weight $w_j \in \mathbb{N}$;
    number of machines $m \in \N$ with $m \geq 1$.}{Compute the minimum total weighted completion time $\sum_{j=1}^n w_j C_j$ of scheduling the given $n$ jobs on $m$ identical parallel machines.}

Whenever the number of machines $m$ is a fixed constant, we denote the problem as \PmSwjCj{m}, e.g. \PmSwjCj{3} denotes the special case of \PSwjCj with $m = 3$.

We solve \PSwjCj using a standard dynamic programming approach as follows.
Each dynamic programming entry tracks the minimum total weighted completion time for a given subset of jobs and a given number of machines.
Instead of considering the naive transition which iterates through all possible subsets of jobs over a single machine, we restrict our focus to the set of all possible jobs on the machine with the least number of jobs.
This ensures that the transitions are computable in time faster than $\Os(3^n)$.
Another crucial observation is that the base cases corresponding to smaller instances with $m \leq 3$ machines can be computed in time $\Os(m^{n/2})$ (as stated in \cref{thm:PmSumwjCj}).
In the end, we show that the worst-case running time is $\Oh(2.755^n)$.
We refer to \cref{sec:wjCj} for more details of the proof of \cref{thm:SumwjCj}.
For now, we discuss the algorithm for $\PmSwjCj{m}$ for $m \leq 3$.
Similar ideas can be used to improve the running time for $m \in \{4, 5, 6\}$ machines as well.

\subparagraph{Number of machines $m \in \{2, 3\}$.}

To solve \PmSwjCj{m} in time $\Os(m^{n/2})$ for $m \in \{2, 3\}$, we use a meet-in-the-middle approach, along with the following data structure due to Guibas et al.~\cite{DBLP:journals/tcs/GuibasSC87}, which is capable of supporting efficient linear programming queries.\footnote{The data structure of Guibas et al.~\cite{DBLP:journals/tcs/GuibasSC87} preprocesses a convex polyhedron (the intersection of $n$ given half-spaces). We obtain the required polyhedron by computing the convex hull of the $n$ input points in time $O(n \log n)$~\cite{10.1145/359423.359430,DBLP:journals/dcg/Chan96}.}

\begin{restatable}[Linear Programming Queries~{\cite{DBLP:journals/tcs/GuibasSC87}}]{lemma}{lpQueries}\label{lem:lp_query}
    Let $d \in \{2, 3\}$ and $\mathcal{P} \subseteq \Z^d$ be a set of size $n$.
    A data structure can be constructed in time $\Oh(n \log{n})$ that answers
    each following query in time $\Oh(\log{n})$:
    Given $(q_1, \dots, q_{d}) \in \Z^{d}$,
    compute $\min_{(r_1, \dots, r_d) \in \mathcal{P}} \sum_{i \in [d]} q_i \cdot r_i$.
\end{restatable}
We split the jobs of the instance into two halves, $A = \{1, \ldots, \floor{n/2}\}$, and $B = \{\floor{n/2} + 1, \ldots, n\}$.
Using \cref{lem:lp_query}, we create a data structure $\Dd$ on $d=m$ dimensions, which contains useful information for every partition $Y_1 \uplus \cdots \uplus Y_m = B$. Preprocessing of this data structure takes time $\Os(m^{n/2})$. Next, for every partition $X_1 \uplus \cdots \uplus X_m = A$, we query $\Dd$ to obtain the minimum weighted completion time when $X_i \cup Y_i$ are the jobs scheduled on the $i$-th machine, over all partitions $Y_1 \uplus \cdots \uplus Y_m = B$. Together, all queries take time $\Os(m^{n/2})$. Finally, we combine all queries to obtain the minimum weighted completion time over all schedules. In total, this approach takes time $\Os(m^{n/2})$. Refer to \cref{lem:mitm-wjCj} in \cref{sec:wjCj} for more details.

\subparagraph{Number of machines $m = 4$.}
From the previous discussion about improving the transitions in the dynamic program, we can use the algorithm for $m \in \{2,3\}$ to obtain an algorithm solving \PmSwjCj{4} in time $\Oh(2.755^n)$ (see \cref{thm:SumwjCj}).
We improve this running time further.
Although the meet-in-the-middle approach we used for $m \in \{2, 3\}$ still works for $m \ge 4$
using known data structures analogous to that of \Cref{lem:lp_query}, albeit with worse running times, (see, e.g., \cite[Corollary 3.2]{DBLP:conf/compgeom/Chan96}), we obtain significantly better running times with a different approach.
A clever trick helps us to reduce $4$-dimensional queries to $3$-dimensional queries, letting us exploit \Cref{lem:lp_query}. While this results in an extra preprocessing overhead, it nevertheless allows us to solve \PmSwjCj{4} even faster.
Interestingly, this extra preprocessing overhead causes the algorithm to achieve a better running time when we split the set of jobs into unequal parts, as opposed to the traditional meet-in-the-middle approach which splits into equal parts.

To capture the split into unequal parts consider a fraction $\alpha \in (0,1)$ which will be set later, let $A$ be the first $\lfloor \alpha n\rfloor$ of the jobs, and $B$ be the remaining.
We guess a partition $X_1 \uplus X_2 \uplus X_3 \uplus X_4 = A$.
Without loss of generality assume $|X_1| \le |X_2| \le |X_3| \le |X_4|$. We create a separate linear programming query data structure (as in \cref{lem:lp_query}) for \emph{every} possible $X_1$, which contains relevant information for every partition $Y_1 \uplus Y_2 \uplus Y_3 \uplus Y_4 = B$.
Note that this steps differs from the algorithm for $m \in \{2, 3\}$ which used a single data structure. Moreover, since we create a separate data structure for every $X_1$, it suffices to set these up on $d=m-1=3$ dimensions.
We show that preprocessing all data structures takes time $\Os(1.755^{|A|} \cdot 4^{|B|})$. Next, for all partitions $X_1 \uplus X_2 \uplus X_3 \uplus X_4 = A$, we query the data structure specific to $X_1$ for the minimum weighted completion time when $X_i \cup Y_i$ are the jobs scheduled on the $i$-th machine, over all partitions $Y_1 \uplus Y_2 \uplus Y_3 \uplus Y_4 = B$.
Together, all queries take $\Os(4^{|A|})$ time. Finally, we combine all queries to obtain the minimum weighted completion time over all schedules.
In total this takes time $\Os(1.755^{|A|} \cdot 4^{|B|} + 4^{|A|})$ which for $\alpha = 0.628$ turns out to be $\Oh(2.389^n)$.
Refer to \cref{lem:P4wjCj} for more details.

\subparagraph{Number of machines $m \in \{5,6\}$.}
To solve \PmSwjCj{m} with $m \in \{5, 6\}$ faster than $\Oh(2.755^n)$, we revisit the dynamic programming approach used to prove \cref{thm:SumwjCj}. Each entry, $\DP[i, S]$, tracks the minimum total weighted completion time for a given subset $S \subseteq \{1, \ldots n\}$ of jobs when scheduled on $i$ machines. Note that the value of $\DP[m, \{1, \ldots, n\}]$ contains the minimum weighted completion time for the original input instance.
We do the following.

\begin{itemize}
    \item For $m = 5$, we compute $\DP[2, S]$ for all subsets $S \subseteq \{1, \ldots, n\}$ and compute $\DP[3, S]$ only for subsets $S \subseteq \{1, \ldots, n\}$ with $|S| \le 3n/5$ using the above algorithms for \PmSwjCj{2} and \PmSwjCj{3}. We show that these computed entries are sufficient to directly compute the value of $\DP[5, \{1, \ldots, n\}]$. This algorithm takes time $\Oh(2.726^n)$, with the bottleneck being the computation of $\DP[3, S]$ for all subsets $S \subseteq \{1, \ldots, n\}$ with $|S| \le 3n/5$. Refer to \cref{lem:P5wjCj} for more details.
    \item For $m = 6$, we compute $\DP[3, S]$ for all subsets $S \subseteq \{1, \ldots, n\}$ by solving \PmSwjCj{3} on $S$. We show that these computed entries are sufficient to directly compute the value of $\DP[6, \{1, \ldots, n\}]$. This algorithm takes time $\Oh(2.733^n)$, with the bottleneck being the computation of $\DP[3, S]$ for all subsets $S \subseteq \{1, \ldots, n\}$. Refer to \cref{lem:P6wjCj} for more details.
\end{itemize}

\subsection{Weighted number of tardy jobs}

\PSwjUj denotes the scheduling problem where the objective is to minimize the total weighted number of tardy jobs.

\defproblem{\PSwjUj}{$n$ jobs where the $j$-th job consists of processing time $p_j\in \N$ and weight $w_j\in \N$ and due date $d_j \in \mathbb{N}$; number of machines $m \in \N$ with $m \geq 1$.}{Compute the minimum total weight of tardy jobs $\sum_{ j=1}^n w_j U_j$  of scheduling the given $n$ jobs on $m$ identical parallel machines.}

To solve \PSwjUj in time $\Os(2^n)$, we iterate over all subsets $S \subseteq [n]$, check whether $S$ can be scheduled on $m$ machines with every job meeting its deadline, and output the minimum possible value of $w([n] \setminus S)$. 
To check whether a subset $S$ can be scheduled on $m$ machines with every job meeting its deadline, we use dynamic programming to maintain whether a subset $S \subseteq [n]$ on $i \le m$ machines can be scheduled with every job meeting its deadline. 
The base case of the dynamic program (i.e. for $i = 1$) can be computed for each subset $S \subseteq [n]$ using a polynomial-time algorithm due to Moore~\cite{moore1968n}, and the transitions together for a fixed $i \ge 2$ can be computed using a Fast Subset Convolution algorithm due to Bj\"{o}rklund et al.~\cite{bjorklund2007fourier}. 
For more details refer to the proof of \cref{thm:SumwjUj} in Section~\ref{sec:wjUj}.

\subsection{Bin Packing}
Let us now give a high level overview of our result for the \binpacking problem.
We solve the decision version of \binpacking as constructing a solution can be done by a standard reduction to polynomially many calls to a decision oracle.

\defproblem{\binpacking}{Multiset $I \subseteq \N$ of $n$ items; number of bins $m \in \N$ with $m \geq 1$; capacity $t \in \N$.}{Decide if there is a partition $X_1 \uplus \dots \uplus X_m$ of $I$ into $m$ parts such that $\sum(X_i) \leq t$ for all $i \in [m]$.}

To solve \binpacking in faster than $2^n$-time, we will make use of the following two important results. First, we recall that Nederlof et al.~\cite{nederlof2023faster} showed that if the number of bins $m$ is a fixed constant (e.g.~$m=6$), then \binpacking can be solved in time faster than~$\Os(2^n)$.

\begin{restatable}[{\cite[Theorem 1.1]{nederlof2023faster}}]{lemma}{lemmaBinPackingMFixed}\label{thm:nederlof_binpacking}
    For every $m \in \mathbb{N}$ there is a constant $\eps_m >0$ such that \binpacking on instances with exactly $m$ bins can be solved with high probability in time $\Os(2^{(1-\eps_m)n})$. The algorithm has no false positives. 
\end{restatable}

Note that \cite[Theorem 1.1]{nederlof2023faster} does not state that the algorithm has no false positives, but this is indeed the case.
Next, since we assume the Asymptotic Rank Conjecture, we can leverage the fast algorithm of Björklund et al.~\cite{radu25} that solves the \threepartition problem as defined below.
For a universe $U$ and a family $\Ff$ of subsets of $U$, we say that $\Ff$ is $\nu$-bounded if no set in $\Ff$ is larger than $\nu|U|$.

\defproblem{\threepartition for $\nu$-bounded set families}{A universe $U \subseteq \N$; $\nu$-bounded set families $\Ff_1, \Ff_2, \Ff_3 \subseteq 2^U$.}{Decide whether there exist pairwise disjoint subsets $S_1 \in \Ff_1$, $S_2 \in \Ff_2$ and $S_3 \in \Ff_3$ such that $S_1 \uplus S_2 \uplus S_3 = U$.}

\begin{restatable}[{\cite[Theorem 1]{radu25}}]{lemma}{lemmaFastThreePartitionARC}\label{thm:arc_3partition}
    Assuming the Asymptotic Rank Conjecture, for every $\eps > 0$, there exists $\eta_\eps >0$ such that the \threepartition problem over an $n$-element universe for $(1-\eps)/2$-bounded set families can be solved deterministically in time $\Os(2^{(1-\eta_{\eps})n})$.
\end{restatable}

The general idea to prove \cref{thm:bin-packing-arc} is to reduce solving the given \binpacking instance either to solving a \binpacking instance with a constant number of bins, which can be solved efficiently using \cref{thm:nederlof_binpacking}, or to solving a \threepartition instance over an $n$-element universe for $(1-\eps)/2$-bounded set families, which can be solved efficiently using \cref{thm:arc_3partition} when assuming ARC.
We now give a high-level overview of this strategy.

Consider a \binpacking instance $(m, t, I)$. Any partition of $I$ into $m$ parts $X_1, \dots, X_m$ is called a \emph{packing} of $I$ into $m$ bins. The quantity $\Sigma(X_i)$ is called the \emph{load of the $i$-th bin}. The packing forms a \emph{solution} if $\Sigma(X_i) \leq t$ for all $i \in [m]$, i.e.~the load of each bin does not exceed the capacity $t$. The goal of the \binpacking problem is to decide whether the given instance admits a solution.
To do that, we will distinguish between four types of solutions. For each solution type, we will show a faster than $2^n$-time algorithm that, given a \binpacking instance, outputs \textsc{No} if the instance has no solution and outputs \textsc{Yes} if the instance admits a solution of that solution type. In the end, given any \binpacking instance, we can simply run the four algorithms for each solution type and output \textsc{No} if and only if all the algorithms output \textsc{No}.

We now define the different solution types we consider. For that, consider a solution $X_1, \dots, X_m$ to that given \binpacking instance $(m, t, I)$. Let $b_i \coloneq |X_i|$ for all $i \in [m]$. Without loss of generality, we can assume that $b_1 \geq \dots \geq b_m$. Let $\eps >0$ be a small enough constant.
Then, we distinguish the following cases.
\begin{caseenum}
    \item If $b_1 \geq (1+\eps)\cdot n / 2$, i.e.~\emph{more than half of the items are packed into the first bin}, then we can guess $\leq (1-\eps)n/2$ items and use standard dynamic programming to verify that they can be packed into $m-1$ bins. 
    See \cref{lem:case_DP} for details. 

    \item If $\sum_{i=1}^6 b_i \geq (1-\eps)\cdot n$, i.e.~\emph{almost all items are packed into the first 6 bins},
    then we can guess $\leq \eps n$ items and use standard dynamic programming to verify that they can be packed into $m-6$ bins, and then use \cref{thm:nederlof_binpacking} to verify that the $\geq (1-\eps)n$ items can be packed into 6 bins. 
    See \cref{lem:case_cstBP} for details. 

    \item If $b_1 \in [\frac{n}{2} \pm \eps\cdot \frac{n}{2}]$ and $\sum_{i=1}^6 b_i < (1-\eps) \cdot n$, i.e.~\emph{the first bin contains approximately half of the items, but the remaining items are not concentrated on the first 6 bins}, then, in \cref{lem:case_balanced}, 
    we use a randomized strategy that samples witnesses of the existence of such a packing.
    This algorithm is inspired by the proof of \cite[Lemma~3.5]{nederlof2023faster}.

    \item If $b_1 \leq (1-\eps)\cdot n / 2$ and $\sum_{i=1}^6 b_i < (1-\eps) \cdot n$, i.e.~\emph{the first bin contains less than half of the items, and the items are not concentrated on the first 6 bins}, we reduce solving the problem to solving specific instances of \threepartition{}, which we then solve using \cref{thm:arc_3partition}. This is the only case in which we make use of the ARC. Refer to \cref{lem:case_ARC} for the details.
\end{caseenum}
Of course, when given a yes-instance of \binpacking{}, we do not know in which case we are.
Hence, we simply apply the algorithms for all the four cases in sequence, one of which will detect that the instance is a yes-instance (with high probability).
 \section{Preliminaries}\label{sec:preliminaries}

In this paper, we work in a unit-cost RAM model with a word size of $\Theta(L)$ bits, where $L$ is the encoding length of the input.
Thus, arithmetic operations on integers with $O(L)$ bits take constant time.
The unit-cost assumption can be dropped at the cost of an additional factor in the running time that is polynomial in $L$.

\paragraph*{Notation}
We use $\nat$ to denote the set of natural numbers, and $0 \in \nat$.
For $n, \delta \in \N$, let $[n] \coloneq \{1, 2, \ldots, n\}$ and $[n \pm \delta] \coloneq \{n-\delta, \dots, n + \delta\}$.
We use $\log$ to denote the binary logarithm $\log_2$, and $\ln$ to denote the natural logarithm with base $e$.
For $x \in \mathbb R$, we adopt the convention that $\infty \cdot x$ is equal to $0$ if $x = 0$, to $\infty$ if $x>0$, and to $-\infty$ if $x <0$.
By convention, the minimum over an empty set is defined to be $\infty$.

We use the $\Os$-notation to hide polynomial factors in the following way: If function $f(n)$ is in $\Os(g(n))$, then $f(n) \in \Oh(g(n) \cdot n^c)$ for some constant $c$.

The Iverson bracket $\llbracket E \rrbracket$ is used to denote the indicator variable of predicate $E$, i.e.~it takes value 1 if $E$ is true, and value 0 otherwise.

For a set $X$, we say that $X_1, \dots, X_p \subset X$ form a partition of $X$ into $p$ parts if $\bigcup_{i \in [p]} X_i = X ,$ and for any $i, j \in [p]$, $i \neq j$ we have $X_i \cap X_j = \emptyset$. Note that in this definition we allow (even multiple) parts $X_i$ to be the empty set.
For a multiset of integers $Y$, we write $\Sigma(Y) \coloneq \sum_{y \in Y} y$.
We say that $Y_1, \dots, Y_p \subseteq Y$ is a partition of $Y$ into $p$ parts if for any fixed ordering of the elements in $Y = \{y_1, \dots, y_n\}$, the sets $X_1, \dots, X_p$ defined as $X_i \coloneq \{j \ \mid \ y_j \in Y_i\}$ for all $i \in [p]$ form a partition of $[n]$.
We use the shorthand notation $Y = Y_1 \uplus \dots \uplus Y_p$ to denote that $Y_1, \dots, Y_p$ is a partition of $Y$ into $p$ parts.

For a collection $\Ww$ of submultisets of some universe $U$,  we define the upward closure of $\Ww$ as $\uc \Ww \coloneq \{W' \supseteq W \mid W \in \Ww, W' \subseteq U\}$,
and the downward closure of $\Ww$ as $\dc \Ww \coloneq \{W' \subseteq W \mid W \in \Ww\}$.

In the context of \PSwjCj and \PSwjUj, job $j$ has processing time $p_j$ and weight $w_j$. We denote $p(S) \coloneq \sum_{j \in S}p_j$ and $w(S) \coloneq \sum_{j \in S}w_j$ for any $S \subseteq [n]$.

In the context of \binpacking algorithms, the term \emph{with high probability}
means a probability of at least $1 - n^{-\Omega(1)}$, where $n$ is the number of
items.

\paragraph*{Binomial Coefficients}

For a (multi)set $S$ and non-negative integer $i \leq |S|$, the set of sub(multi)sets of $S$ of size $i$ is denoted by $\binom{S}{i}$.
Further, for non-negative real number $x$, let $\binom{S}{\le x} := \bigcup _{j = 0}^{\lfloor x \rfloor} \binom{S}{j}$.
For integer $n$ and non-negative real number $r$ we additionally define $\binom{n}{\le r} \coloneq \sum_{i = 0}^{\lfloor r \rfloor} \binom{n}{i}$.

We use the following bound for $\alpha \in [0, 1]$ (see~\cite[Page 353]{thomas2006elements}):
\[\binom{n}{\alpha n} \leq 2^{H(\alpha)n}\]
where $H(\cdot)$ is the entropy function defined as $ H(\alpha) \coloneq -\alpha \log_2(\alpha) - (1-\alpha) \log_2(1-\alpha)$ for $\alpha \in (0, 1)$ and $H(\alpha)=0$ for $\alpha \in \{0, 1\}$.
Note that $H(\alpha) \in (0, 1)$ for $\alpha \in (0,1) \setminus \{1/2\}$, so $2^{H(\alpha)} \in (1, 2)$.
Hence, by setting $\eps_\alpha \coloneq 2-2^{H(\alpha)} \in (0, 1)$ for $\alpha \in (0, 1) \setminus \{1/2\}$, we can rewrite the above bound as
\[\binom{n}{\alpha n} \leq (2-\eps_\alpha)^n.\]

Plugging in the value $\alpha = 1/4$, and by leveraging the fact that $\binom{n}{r_1} \le \binom{n}{r_2}$ when $r_1 \le r_2 \le n/2$, we obtain the following bound:
\begin{equation}\label{prop:binom-to-entropy}
    \binom{n}{\le n/4} \le \Os(2^{H(\frac{1}{4})n}) \le \Oh(1.7548^n).
\end{equation}

Note that in \cref{prop:binom-to-entropy} we round up the exponent base to $4$ digits after the decimal point. This rounding up additionally allows us to replace $\Os(\cdot)$ with $\Oh(\cdot)$.

\paragraph*{Technical Lemmas}

Now, we cover some technical lemmas we utilize in the proofs.
First, we restate the following combinatorial bound established by Nederlof et al.~\cite{nederlof2023faster}.

\begin{lemma}[{\cite[Lemma B.2]{nederlof2023faster} with $c=0$}]\label{lem:nederlof2023_bound}
    Let $I \subset \N$ be a multiset of $n$ integers and $z \in (0, 1)$.
    Then for any $\Ww \subseteq \binom{I}{n/2}$ of size at most $2^{(1-z)n}\poly(n)$, we have
    \[
        |\dc \Ww| + |\uc \Ww| \leq \Oh(2^{(1-\rho)n})
    \]
    where $\rho = \frac{2}{\ln 2} \left(\frac{z}{4 \log(12/z)}\right)^2$.
\end{lemma}

Note that \cite{nederlof2023faster} state the lemma for sets only, but we can easily obtain the statement for multisets by applying it on the index set.
We also require the following bound.

\begin{lemma}\label{lem:0.6-binom}
    We have $\sum \limits_{r = 0}^{\lfloor 0.6n \rfloor} \binom{n}{r} \cdot \left(\sqrt{3}\right)^{r} \le \Oh(2.726^n)$.
\end{lemma}
\begin{proof}
    Consider the function $f(x) : (0, 0.6] \to \mathbb{R}$ defined as
    \[f(x) = {x(\ln \sqrt{3} - \ln x) - (1 - x) \ln (1 - x)}.\]
    Therefore, for $r \in [\lfloor 0.6n \rfloor]$, we have
    \begin{align*}
        \binom{n}{r} \cdot \left(\sqrt{3}\right)^{r}
         & \le 2^{H(r/n) n} \cdot 2^{r \log {\sqrt{3}}}                                                      \\
         & = 2^{((r/n)(\log \sqrt{3} - \log (r/n)) - (1 - (r/n)) \log (1 - (r/n)))n}                         \\
         & = \left( \exp \left((r/n)(\ln \sqrt{3} - \ln (r/n)) - (1 - (r/n)) \ln (1 - (r/n))\right)\right)^n \\
         & = \left(e^{f(r/n)}\right)^n.
    \end{align*}

    Consider the first derivative $f'(x)$ of $f(x)$,
    \begin{align*}
        f'(x) & = \frac{d}{dx} \left(x(\ln \sqrt{3} - \ln x) - (1 - x) \ln (1 - x)\right)                                 \\
              & = (\ln \sqrt{3} - \ln x) + x \left(- \frac{1}{x}\right) - (- \ln (1-x)) - \left( - \frac{1-x}{1-x}\right) \\
              & = \ln \sqrt{3} - \ln x + \ln (1-x).
    \end{align*}
    Now, consider the second derivative $f''(x)$ of $f(x)$,
    \begin{align*}
        f''(x) & = \frac{d}{dx} \left(\ln \sqrt{3} - \ln x + \ln (1-x)\right) = - \frac{1}{x} - \frac{1}{1-x} < 0 & \text{for }x\in(0, 0.6).
    \end{align*}
    Thus $f'(x)$ is decreasing in the range $(0, 0.6]$. In particular, for all $x \in (0, 0.6)$, we have,
    \[f'(x) > f'(0.6) = \ln \sqrt{3} - \ln 0.6 + \ln 0.4 = \ln \left(\frac{\sqrt{3} \cdot 0.4}{0.6}\right) = \ln \left(\frac{2\sqrt{3}}{3}\right) > 0.\]
    This in turn implies $f(x)$ is increasing in the range $(0, 0.6]$. Plugging in the exact values gives us $f(0.6) \le 1.0026$ and  $e^{f(0.6)} \le 2.7254$. Therefore, for all $r \in [\lfloor 0.6n \rfloor]$, we have,
    \begin{align*}
        \binom{n}{r} \cdot \left(\sqrt{3}\right)^{r} \le \left(e^{f(r/n)}\right)^n \le \left(e^{f(0.6)}\right)^n \le 2.7254^n.
    \end{align*}
    This implies,
    \[\sum \limits_{r = 0}^{\lfloor 0.6n \rfloor} \binom{n}{r} \cdot \left(\sqrt{3}\right)^{r} = 1 + \sum \limits_{r = 1}^{\lfloor 0.6n \rfloor} \binom{n}{r}\left(\sqrt{3}\right)^{r} \le 1 + \lfloor 0.6n \rfloor \cdot 2.7254^n = \Oh(2.726^n).\]
    This completes the proof.
\end{proof}

Now, we prove the following general result which can be applied to \PSwjCj or \PSwjUj by setting $f_j$ to $C_j$ or $U_j$ respectively.

\begin{restatable}{lemma}{OPTrecursion}\label{lem:dp-correct-recursion}
    Let \PSwjfj be a scheduling problem on $m$ parallel identical machines and $n$ jobs where the objective is to minimize $\sum_{j=1}^n w_j f_j$, where $f_j$ is a cost function depending on the job $j$ and its completion time $C_j$. For $i \in [m]$ and $S \subset [n]$, let $\OPT(i, S)$ denote the minimum value of $\sum_{j \in S} w_jf_j$ for the instance restricted to $i$ machines and the jobs in $S$.

    Then, for all $i \in \{2, 3, \ldots, m\}$, $S \subseteq [n]$, and $j \in [i-1]$, the following holds true:
    \[\OPT(i, S) = \min \limits_{S' \subseteq S}\{\OPT(j, S') + \OPT(i - j, S \setminus S')\} = \min \limits_{S' \in \binom{S}{\le j|S|/i}}\{\OPT(j, S') + \OPT(i - j, S \setminus S')\}.\]
\end{restatable}

\begin{proof}
    Consider an optimal schedule of the jobs $S$ on $i$ machines. Let $S^* \subset S$ be the jobs scheduled on the $j$ machines containing the least number of total jobs; therefore $|S^*| \le j|S|/i$. Since the objective is the weighted sum of the cost functions over all jobs, $\sum_{j' \in [S]}w_{j'}f_{j'}$, the jobs $S^*$ must be scheduled optimally on $j$ machines, and the jobs $S \setminus S^*$ must be also scheduled optimally on $i - j$ machines. This gives us $\OPT(i,S) = \OPT(j, S^*) + \OPT(i - j, S \setminus S^*)$. Moreover, as $S^* \in \binom{S}{\le j|S|/i}$, we have
    \begin{equation}\label{eq:opt-breakup-jS_by_i}
        \OPT(i,S) = \OPT(j, S^*) + \OPT(i - j, S \setminus S^*) \ge \min \limits_{S' \in \binom{S}{\le j|S|/i}}\{\OPT(j, S') + \OPT(i - j, S \setminus S')\}.
    \end{equation}
    Next, consider a subset $S^\dagger \in \argmin_{S' \subseteq S}\{\OPT(j, S') + \OPT(i - j, S \setminus S')\}$.
    One way of scheduling the jobs $S$ on $i$ machines is scheduling the jobs $S^\dagger$ on $j$ machines optimally, and scheduling the jobs $S \setminus S^\dagger$ on $i - j$ machines optimally.
    Such a schedule has the weighted sum of cost functions equal to $\OPT(j, S^\dagger) + \OPT(i - j, S \setminus S^\dagger)$ and this can not be less than $\OPT(i, S)$. This gives us,
    \begin{equation}\label{eq:opt-breakup-all}
        \OPT(i,S) \le \OPT(j, S^\dagger) + \OPT(i - j, S \setminus S^\dagger) = \min \limits_{S' \subseteq S}\{\OPT(j, S') + \OPT(i - j, S \setminus S')\}.
    \end{equation}
    Finally, since $S' \in \binom{S}{\le j|S|/i}$ implies $S' \subseteq S$, we have
    \begin{equation}\label{eq:trivial-min-ineq}
        \min \limits_{S' \in \binom{S}{\le j|S|/i}}\{\OPT(j, S') + \OPT(i - j, S \setminus S')\} \ge \min \limits_{S' \subseteq S}\{\OPT(j, S') + \OPT(i - j, S \setminus S')\}.
    \end{equation}
    Combining \cref{eq:opt-breakup-jS_by_i,eq:opt-breakup-all,eq:trivial-min-ineq}, we get
    \begin{align*}
        \OPT(i,S) & \le \min \limits_{S' \subseteq S}\{\OPT(j, S') + \OPT(i - j, S \setminus S') \}                              \\
                  & \le \min \limits_{S' \in \binom{S}{\le j|S|/i}}\{\OPT(j, S') + \OPT(i - j, S \setminus S')\} \le \OPT(i, S).
    \end{align*}
    Hence, equality must hold throughout, which proves the lemma.
\end{proof}

\paragraph*{Linear Programming Queries}
Finally, we elaborate more on the linear programming query data structure we utilize.
The work of Guibas et al.~\cite{DBLP:journals/tcs/GuibasSC87} gives a data structure that takes $n$ half-spaces in $d \in \{2,3\}$ dimensions as input.
As part of the preprocessing, first the polyhedron $P$ that is the intersection of these $n$ half-spaces is computed using an algorithm of Preparata and Muller~\cite{preparataFindingIntersectionHalfspaces1979}. 
Then this polyhedron $P$ is processed further, note that the rest of the preprocessing does not utilize the half-spaces and only requires $P$.
Overall, the preprocessing takes time $O(n \log n)$.
Given any point $w$ in $d$-dimensional space, the data structure can find $\max_{p \in P} w \cdot p$ in time $O(\log n)$.
This is exactly the same as solving a linear program with $d$ variables where the cost function is given by $w$, and the feasible region is defined by the polyhedron $P$.
Hence, such a data structure is useful when the constraints of the linear program remain stable, but the objective function changes often.

We can use the same data structure for our purposes.
Let $\mathcal{P}$ be a set of $n$ points, then we can compute the convex hull, say $P$, of these points in time $O(n \log n)$ using well-known algorithms~\cite{10.1145/359423.359430,DBLP:journals/dcg/Chan96}.
Then, instead of taking $n$ half-spaces as input and computing their intersection, the data structure can just work with this polyhedron $P$.
Hence, given a point $w$, we can find $-\max_{p \in P} -w \cdot p = \min_{p \in P} w \cdot p$ in time $O(\log n)$.
Since $P$ is bounded, there is always an extreme point $p$ of $P$ that minimizes $w \cdot p$ among all points of $P$, and therefore $p \in \mathcal{P}$.
This directly gives us the required data structure, we restate the lemma that summarizes this discussion for convenience.

\lpQueries*{} \section{Weighted completion time}\label{sec:wjCj}

In this section we focus on the problem of minimizing the total weighted completion time \PSwjCj and prove \cref{thm:SumwjCj,thm:PmSumwjCj}.

\subsection{The Proof of Theorem~\ref{thm:SumwjCj}}
We begin with the proof of \cref{thm:SumwjCj} that we restate for convenience.

\SumwjCj*

We prove \cref{thm:SumwjCj} using dynamic programming.
The key observation is that the base cases, computing the dynamic programming tables for $m\leq 3$ machines, can be handled efficiently. 
Therefore, we first show how to solve \PmSwjCj{m} on $m \in \{1, 2, 3\}$ machines, before presenting the dynamic program used to obtain \cref{thm:SumwjCj}. 
This is done by exploiting the meet-in-the-middle paradigm combined with an efficient data structure that supports linear programming queries in $m$ dimensions.

Recall that for $m=1$ machine, the following observation due to Smith~\cite{smith1956various} implies that the problem can be solved by simply sorting the jobs.
Additionally, we assume without loss of generality that the given jobs are sorted by their ratio of processing time to weight $p_j/w_j$. 
If $w_j = 0$ for job $j$ we assign $p_j/w_j = \infty$.

\begin{lemma}[Smith's rule~\cite{smith1956various}]\label{lem:wjCj-jobs-sorted}
    In an optimal schedule for an instance of \PSwjCj, jobs on the same machine are scheduled in non-decreasing order of their ratio of processing time to weights.
\end{lemma}

To efficiently solve \PmSwjCj{m} for $m \in \{2, 3\}$, we use the meet-in-the-middle paradigm. 
For $X \subseteq [n]$, let $\cost(X)$ be the minimum weighted completion time for scheduling the jobs $X$ on a single machine. 
In the following lemma, we show how to combine the cost of two disjoint subset of jobs.

\begin{lemma}\label{lem:mitm-cost-breakup}
    Consider an instance of \PSwjCj where the $n$ jobs satisfy $\frac{p_1}{w_1} \le \frac{p_2}{w_2} \le \cdots \le \frac{p_n}{w_n}$. Let $a \in [n]$, $A = [a]$, and $B = [n] \setminus A$. Consider partitions $X_1 \uplus \cdots \uplus X_m = A$ and $Y_1 \uplus \cdots \uplus Y_m = B$.
    Then, the minimum weighted completion time of a schedule where the jobs $X_i \cup Y_i$ are scheduled on the $i$-th machine for $i \in [m]$, is given by
    \[\sum \limits_{i \in [m]} \cost(X_i \cup Y_i) = \sum \limits_{i \in [m]}\cost(X_i) + z_0 + \sum_{i \in [m-1]}z_i \cdot p(X_i),\]
    where $z_0 = \sum_{i \in [m]}\cost(Y_i) + \left(w(B) - \sum_{r \in [m-1]} w(Y_r)\right) \cdot p(A)$, and for every $i \in [m-1]$, $z_i = w(Y_i) - w(B) + \sum_{r \in [m - 1]}w(Y_r)$.
\end{lemma}
\begin{proof}
    Consider the optimal schedule that processes jobs $X_i \cup Y_i$ on the $i$-th machine. 
    Then by Smith's rule (\cref{lem:wjCj-jobs-sorted}), the $i$-th machine first processes jobs in $X_i$ in order of their indices, and then processes jobs in $Y_i$ in order of their indices. 

    Let $C_j$ be the completion time of job $j \in [n]$ in the optimal schedule and denote by $C^{(i)}_j \coloneq \sum \limits_ {k \in Y_i, k \leq j} p_{k}$ the completion time of job $j \in Y_i$ in an optimal schedule of $Y_i$ on a single machine. 
    In particular, $\cost(Y_i) = \sum_{j \in Y_i} w_j C^{(i)}_j$. 
    Since the jobs in $X_i$ complete at time $p(X_i)$, we have $C_j = p(X_i) + C^{(i)}_j$ for every job $j \in Y_i$. 
    Therefore, the total cost of scheduling $X_i \cup Y_i$ on a single machine amounts to
    \begin{align*}
        \cost(X_i \cup Y_i) &= \cost(X_i) + \sum \limits_{j \in Y_i} w_j \left(C^{(i)}_j + p(X_i)\right) \\
                            &= \cost(X_i) + \sum \limits_{j \in Y_i} w_j C^{(i)}_j + p(X_i) \cdot \sum \limits_{j \in Y_i} w_j \\
                            &= \cost(X_i) + \cost(Y_i) + p(X_i) \cdot w(Y_i).
    \end{align*}
    This implies that the minimum weighted completion time of the entire schedule is
    \begin{align*}
        \sum \limits_{i \in [m]}\cost(X_i \cup Y_i) =& \sum \limits_{i \in [m]}\left(\cost(X_i) + \cost(Y_i) + p(X_i) \cdot w(Y_i)\right) \\ 
        =& \sum \limits_{i \in [m]}\cost(X_i) + \sum \limits_{i \in [m]} \cost(Y_i) + \sum \limits_{i \in [m - 1]} p(X_i)w(Y_i) \\
        &\quad + \left(p(A) - \sum \limits_{i \in [m-1]} p(X_i)\right) \cdot \left(w(B) - \sum \limits_{r \in [m-1]} w(Y_r)\right) \\
        =& \sum \limits_{i \in [m]}\cost(X_i) + \sum \limits_{i \in [m - 1]} \left( w(Y_i) - w(B) + \sum_{r \in [m - 1]}w(Y_r) \right) \cdot p(X_i)\\
        &\quad + \sum \limits_{i \in [m]}\cost(Y_i) + \left(w(B) - \sum \limits_{r \in [m-1]} w(Y_r)\right) \cdot p(A) \\
        =& \sum \limits_{i \in [m]}\cost(X_i) + z_0 + \sum_{i \in [m-1]}z_i \cdot p(X_i). \qedhere
    \end{align*}
\end{proof}

Now, we are ready to prove \cref{lem:mitm-wjCj}.
\begin{lemma}\label{lem:mitm-wjCj}
    \PmSwjCj{m} can be solved in $\Os(m^{n/2})$ time, for $m \in \{2, 3\}$.
\end{lemma}
\begin{proof}
    Consider an instance of \PmSwjCj{m} on $m \in \{2, 3\}$ machines and $n$ jobs with processing times $p_1, \dots, p_n$ and weights $w_1, \dots, w_n$. 
    Split the jobs into two halves $A \coloneq [\floor{n/2}]$ and $B \coloneq [n] \setminus A$. 
    Broadly, the algorithm will do the following: 
    guess the optimal partition $X_1 \uplus \cdots \uplus X_m = A$ such that jobs in $X_i$ are processed on the $i$-th machine in the optimal schedule and then use the data structure from \cref{lem:lp_query} to compute the minimum weighted completion time when jobs in $X_i \cup Y_i$ are processed on the $i$-th machine, over all partitions $Y_1 \uplus \cdots \uplus Y_m = B$. The details in the remainder of the proof are presented in the following structure: (i) building a linear representation of the minimum weighted completion time for specific schedules using \cref{lem:mitm-cost-breakup}, (ii) designing an algorithm which leverages this linear representation, (iii) arguing the correctness of the algorithm and (iv) analyzing its running time.
    
    \proofsubparagraph*{Linear representation of the minimum weighted completion time.} Consider partitions $X_1 \uplus \cdots \uplus X_m = A$ and $Y_1 \uplus \cdots \uplus Y_m = B$ such that $X_i 
    \cup Y_i$ is optimally scheduled on the $i$-th machine for $i \in [m]$. By \cref{lem:mitm-cost-breakup}, the minimum weighted completion time of such a schedule is  
    \[ \sum \limits_{i \in [m]} \cost(X_i \cup Y_i) = \sum \limits_{i \in [m]}\cost(X_i) + z_0 + \sum_{i \in [m-1]}z_i \cdot p(X_i),\]
    where $z_0 = \sum_{i \in [m]}\cost(Y_i) + \left(w(B) - \sum_{r \in [m-1]} w(Y_r)\right) \cdot p(A)$, and for every $i \in [m-1]$, $z_i = w(Y_i) - w(B) + \sum_{r \in [m - 1]}w(Y_r)$.
    Notice that $z_0, z_1, \ldots, z_{m-1}$ depend only on $Y_1$, \ldots, $Y_m$, $A$, $B$, but \emph{not} on $X_1$, \dots, $X_m$. 
    Define
    \begin{equation}\label{eq:R-defn}
        \Rr \coloneq \left\{(z_0, z_1, \ldots, z_{m-1}) \mid Y_1 \uplus \cdots \uplus Y_m = B \text{ and }z_0, z_1, \ldots, z_{m-1}\text{ are as defined above}\right\}.
    \end{equation}
    
    Now, for a fixed partition $X_1 \uplus \cdots \uplus X_m = A$ corresponding to a schedule of jobs in $A$ across the machines, we extend such a schedule optimally to also include jobs in $B$. The minimum weighted completion time of this is given by the minimum of $\sum_{i \in [m]}\cost(X_i \cup Y_i)$ over all partitions $Y_1 \uplus \cdots \uplus Y_m = B$. This is equal to
    \begin{align*}
        &\min \limits_{(z_0, z_1, \ldots, z_{m-1}) \in \Rr} \left\{\sum \limits_{i \in [m]}\cost(X_i) + z_0 + \sum_{i \in [m-1]}z_i \cdot p(X_i)\right\}\\
        & \qquad = \sum \limits_{i \in [m]}\cost(X_i) + \min \limits_{(z_0, z_1, \ldots, z_{m-1}) \in \Rr} \left\{\sum_{i = 0}^{m-1}z_i \cdot x_i\right\},
    \end{align*}
    where $x_0 = 1$ and $x_i = p(X_i)$ for all $i \in [m-1]$.

    \proofsubparagraph*{Description of the algorithm.} First, we compute $\Rr$ using \cref{eq:R-defn}.
    Then, we use \cref{lem:lp_query} to set up a data structure $\Dd$ for $m \in \{2,3\}$ dimensions, preprocessed with all elements $(z_0, z_1, \ldots, z_{m-1}) \in \Rr$. Recall that the data structure $\Dd$ is capable of answering linear programming queries of the form $\min_{(z_0, z_1, \ldots, z_{m-1}) \in \Rr} \left\{\sum_{i = 0}^{m-1}z_i \cdot x_i\right\}$ for any query point $(x_0, \ldots, x_{m-1})$.
    We then iterate over all partitions $X_1 \uplus \cdots \uplus X_m = A$, and for each partition we compute the value $\xi = \sum_{i \in [m]}\cost(X_i)$. Next, we set $(x_0 = 1, x_1 = p(X_1), \ldots, x_{m-1} = p(X_{m-1}))$, and query $\Dd$ for the value $q = \min_{(z_0, z_1, \ldots, z_{m-1}) \in \Rr} \left\{\sum_{i = 0}^{m-1}z_i \cdot x_i\right\}$. 
    Finally, we output the minimum value of $(\xi + q)$ over all partitions $X_1 \uplus \cdots \uplus X_m = A$. We present a pseudocode in \cref{alg:mitm-wjCj}.
    
    \begin{algorithm}[!htbp]
        \caption{\PmSwjCj{m} for jobs $i \in [n]$ having weights $w_i$ and processing times $p_i$; $m \in\{2,3\}$}\label{alg:mitm-wjCj}
        \begin{algorithmic}[1]
            \State{Sort jobs $i \in [n]$ in non-decreasing order of $\frac{p_i}{w_i}$.}
            \State{$A \gets [ \floor{n / 2} ]$, $B \gets [n] \setminus A$}
            \State{$\Rr \gets \left\{(z_0, z_1, \ldots, z_{m-1}) \mid Y_1 \uplus \cdots \uplus Y_m = B\right\}$} \Comment{as defined in \cref{eq:R-defn}}
            \State{Use \cref{lem:lp_query} to set up a data structure $\Dd$ with all $(z_0, z_1, \ldots, z_{m-1}) \in \Rr$}
            \State{$\ans \gets \infty$}
            \For{$X_1 \uplus \cdots \uplus X_m = A$}
                \State{$\xi \gets \sum_{i \in [m]}\cost(X_i)$}
                \State{$x_0 \gets 1$, $x_1 \gets p(X_1)$, \ldots, $x_{m-1} \gets p(X_{m-1}$)}
                \State{Query $\Dd$ for $q \gets \min_{(z_0, z_1, \ldots, z_{m-1}) \in \Rr} \left\{\sum_{i = 0}^{m-1}z_i \cdot x_i\right\}$}
                \State{$\ans \gets \min\left(\ans, \xi + q\right)$}
            \EndFor
            \State{\Return $\ans$}
        \end{algorithmic}
    \end{algorithm}

    \proofsubparagraph{Proof of correctness.} Fix a partition $X_1 \uplus \cdots \uplus X_m = A$. The discussion on the mathematical representation of the weighted completion time implies that the computed value $(\xi + q) = \sum_{i \in [m]}\cost(X_i) + \min_{(z_0, z_1, \ldots, z_{m-1}) \in \Rr} \left\{\sum_{i = 0}^{m-1}z_i \cdot x_i\right\}$ is equal to the minimum total weighted completion time of all input jobs, when $X_i$ is scheduled completely on the $i$-th machine. The algorithm outputs the minimum value of $(\xi + q)$ over all choices of $X_1 \uplus \cdots \uplus X_m = A$, and hence correctly outputs the minimum weighted completion time for the input instance.

    \proofsubparagraph{Running-time analysis.} Observe that $|\Rr| \le m^{|B|} \leq m^{n/2 + 1}$, as this bounded by the number of partitions $Y_1 \uplus \cdots \uplus Y_m = B$.
    Moreover, $\Rr$ can be computed in time $\Os(m^{n/2})$. 
    Preprocessing the data structure $\Dd$ with all $(z_0, \ldots, z_{m-1}) \in \Rr$ takes time $\Oh\left(|\Rr|\log|\Rr|\right) = \Os(m^{n/2})$, and any query can be answered in $\Oh\left(\log|\Rr|\right) = {\Oh(n)}$ time (\cref{lem:lp_query}). In total, the number of queries made to $\Dd$ is $\Oh(m^{|A|})$, one for each partition $X_1 \uplus \cdots \uplus X_m = A$. Hence, a total time of $\Oh(n) \cdot \Oh(m^{|A|}) = \Os(m^{n/2})$ is spent over all queries. Therefore, the algorithm takes a total time of $\Os(m^{n/2})$.  
\end{proof}

We are now ready to prove \cref{thm:SumwjCj}. 
\begin{proof}[Proof of \cref{thm:SumwjCj}]
    Consider an instance of \PSwjCj, i.e.~consider $n$ jobs with processing times $p_1, \dots, p_n$ and weights $w_1, \dots, w_n$, and let $m$ be the number of available identical parallel machines. 
    Note that if $m \geq n$, then each job can be processed at time $0$ on a distinct machine so that its completion time is precisely $p_j$. Thus, the optimal weighted completion time is simply $\sum_{i=1}^n w_i p_i$. 
    Hence, suppose that $m < n$. 
    We construct a DP table such that for any subset $S \subseteq [n]$ and $i \in [m]$, the table entry $\DP[i, S]$ contains the minimum total weighted completion time when the jobs in $S$ are scheduled on $i$ machines. 

    For every $S \subseteq [n]$, initialize $\DP[1, S]$ using \cref{lem:wjCj-jobs-sorted}, and initialize $\DP[2, S]$ and $\DP[3, S]$ using \cref{lem:mitm-wjCj}. 
    By \cref{lem:dp-correct-recursion}, used with $j = 1$, the following recurrence holds for any $i \geq 4$ and $S \subset [n]$.
    \[\DP[i, S] = \min \limits_{S' \in \binom{S}{\le |S|/i}} \{ \DP[1, S'] + \DP[i - 1, S \setminus S'] \}.\]
    We can therefore iterate over $i \in \{4, \dots, m\}$ and $S \subset [n]$ in increasing order, and compute all table entries $\DP[i, S]$ using the above recurrence. In the end, we output $\DP[m, [n]]$. 
    
    The computed value of $\DP[m, [n]]$ correctly contains the minimum weighted completion time of scheduling all jobs on $m$ machines by \cref{lem:wjCj-jobs-sorted,lem:mitm-wjCj,lem:dp-correct-recursion} and thus our algorithm is correct.
    To analyze the running time, note that the initialization for $i=1$ takes time $\Os(2^n)$ and by \cref{lem:mitm-wjCj}, the initialization for $i=2$ and $i=3$ takes time at most
    \[\sum \limits_{S \subseteq [n]} \Os\left(3^{|S|/2}\right) = \Os \left(\sum \limits_{r = 0}^{n} \binom{n}{r}\cdot \left(\sqrt{3}\right)^{r}\right) = \Os\left(\left(1+\sqrt{3}\right)^n\right) = \Oh(2.733^n).\]
    Each transition in the above recurrence takes time $\Oh(\binom{S}{\le|S|/i})$ and thus the total time to compute $\DP[i, S]$ for all $S \subseteq [n]$ and $i \in \{4, \dots, m\}$ is at most
    \begin{align*}
        \sum_{i=4}^m \sum \limits_{S \subseteq [n]} \Oh\left(\binom{S}{\le|S|/i}\right) 
        &\le \sum \limits_{S \subseteq [n]} \Os\left(\binom{S}{\le|S|/4}\right) \\
        &\leq \sum \limits_{S \subseteq [n]} \Os\left(1.7548^{|S|}\right) \\
        &= \sum \limits_{r = 0}^{n}\binom{n}{r}\cdot \Os\left(1.7548^r\right) = \Os\left(2.7548^n\right).
    \end{align*}
    The second inequality uses the bound $\binom{|S|}{\le |S|/4} = \Oh(1.7548^{|S|})$ from \cref{prop:binom-to-entropy}. 
    In total, the algorithm runs in time $\Oh(2.733^n) + \Os(2.7548^n) \le \Oh(2.755^n)$.
\end{proof}

\subsection{The Proof of Theorem~\ref{thm:PmSumwjCj}}
Now, we proceed to the proof of \cref{thm:PmSumwjCj} that we restate for convenience.
\PmSumwjCj*

Note that \cref{lem:mitm-wjCj} already proves one of the cases of \cref{thm:PmSumwjCj}. We now prove the remaining cases separately as \cref{lem:P4wjCj,lem:P5wjCj,lem:P6wjCj}, which immediately yields \cref{thm:PmSumwjCj}.

\begin{restatable}{lemma}{PfourwjCj}\label{lem:P4wjCj}
    \PmSwjCj{4} can be solved in $\Oh(2.389^{n})$ time.
\end{restatable}

\begin{proof}
    Consider an instance of \PmSwjCj{4}, where the input consists of $n$ jobs with processing times $p_1, \dots, p_n$ and weights $w_1, \dots, w_n$.
    Let $\alpha \coloneqq 0.628$, the choice of $\alpha$ will become clear later in the proof, and let $A = [\floor{\alpha n}]$, and $B = [n] \setminus A$.
    Broadly, the algorithm does the following.
    Without loss of generality, assume that the first machine schedules the smallest number of jobs from $A$.
    Let $X_1$ be the set of jobs scheduled on the first machine. This implies $|X_1| \le |A|/4$. For every $t \in \{p(X_1) \mid X_1 \subseteq A, |X_1| \le |A|/4\}$, we use \cref{lem:lp_query} to set up a data structure $\Dd_t$ which is preprocessed to contain the relevant information for every partition $Y_1 \uplus Y_2 \uplus Y_3 \uplus Y_4 = B$.
    Next, we iterate over all partitions $X_1 \uplus X_2 \uplus X_3 \uplus X_4$ with $|X_1| \le |X_2| \le |X_3| \le |X_4|$, and query $\Dd_{p(X_1)}$ for the minimum weighted completion time when $X_i \cup Y_i$ are the jobs scheduled on the $i$-th machine, over all partitions $Y_1 \uplus Y_2 \uplus Y_3 \uplus Y_4 = B$. Finally, we output the minimum weighted completion time over all choices of $X_1 \uplus X_2 \uplus X_3 \uplus X_4 = A$. 
    The details in the remainder of the proof are presented in the following structure: (i) building a linear representation of the minimum weighted completion time for specific schedules using \cref{lem:mitm-cost-breakup}, (ii) designing an algorithm which leverages this linear representation, (iii) arguing the correctness of the algorithm and (iv) analyzing its running time.

    \proofsubparagraph*{Linear representation of the minimum weighted completion time.} Consider partitions $X_1 \uplus X_2 \uplus X_3 \uplus X_4 = A$ and $Y_1 \uplus Y_2 \uplus Y_3 \uplus Y_4 = B$ such that $X_i 
    \cup Y_i$ is optimally scheduled on the $i$-th machine for $i \in [4]$. By \cref{lem:mitm-cost-breakup}, the minimum weighted completion time of such a schedule amounts to  
    \begin{align*}
        \sum \limits_{i \in [4]} \cost(X_i \cup Y_i) &= \sum \limits_{i \in [4]}\cost(X_i) + z_0 + \sum_{i \in [3]}z_i \cdot p(X_i)\\
        &=\sum \limits_{i \in [4]}\cost(X_i) + \left((z_0 + z_1 \cdot p(X_1)) + z_2 \cdot p(X_2) + z_3 \cdot p(X_3)\right),
    \end{align*}
    where $z_0 = \sum_{i \in [4]}\cost(Y_i) + \left(w(B) - \sum_{r \in [3]} w(Y_r)\right) \cdot p(A)$, and for every $i \in [3]$, $z_i = w(Y_i) - w(B) + \sum_{r \in [3]}w(Y_r)$.
    Let $t = p(X_1)$. This gives us 
    \begin{align*}
        \sum \limits_{i \in [4]} \cost(X_i \cup Y_i) &= \sum \limits_{i \in [4]}\cost(X_i) + \left((z_0 + z_1 \cdot t) + z_2 \cdot p(X_2) + z_3 \cdot p(X_3)\right)\\
        &= \sum \limits_{i \in [4]}\cost(X_i) + \left(\sum_{i = 1}^{3}z'_i \cdot x'_i\right),
    \end{align*}
    where $z'_1 = z_0 + z_1\cdot t$, $z'_2 = z_2$, $z'_3 = z_3$ and $x'_1 = 1$, $x'_2 = p(X_2)$, $x'_3 = p(X_3)$. Notice that $z'_1, z'_2, z'_3$ depend on $Y_1$, $Y_2$, $Y_3$, $Y_4$, $t = p(X_1)$, $A$, $B$; but \emph{not} on $X_2$, $X_3$, $X_4$.
    
    Without loss of generality, assume $|X_1| \le |X_2| \le |X_3| \le |X_4|$. Therefore, $|X_1| \le |A|/4$. Let $\Zz = \{p(X) \mid X \subseteq A, |X| \le |A|/4\}$ and thus $t = p(X_1) \in \Zz$. 
    For every $t \in \Zz$, define   
    \begin{equation}\label{eq:Rt-defn}
        \Rr_t = \left\{(z'_1, z'_2, z'_3) \mid Y_1 \uplus Y_2 \uplus Y_3 \uplus Y_4 = B, \text{ and }z'_1,z'_2,z'_3\text{ as defined above}\right\}.
    \end{equation}
    
    Now, for a fixed partition $X_1 \uplus X_2 \uplus X_3 \uplus X_4 = A$ with $|X_1| \le |X_2| \le |X_3| \le |X_4|$, corresponding to a schedule of jobs in $A$ across the machines, we extend such a schedule optimally to also include jobs in $B$. Let $t = p(X_1)$. The minimum total weighted completion time of this is given by the minimum of $\sum_{i \in [4]}\cost(X_i \cup Y_i)$ over all partitions $Y_1 \uplus Y_2 \uplus Y_3 \uplus Y_4 = B$. This is equal to
    \[\min \limits_{(z'_1, z'_2, z'_3) \in \Rr_t} \left\{\sum \limits_{i \in [4]}\cost(X_i) + \sum_{i=1}^{3} z'_i \cdot x'_i\right\}=\sum \limits_{i \in [4]}\cost(X_i) + \min \limits_{(z'_1, z'_2, z'_3) \in \Rr_t} \left\{\sum_{i=1}^{3} z'_i \cdot x'_i\right\}\]
    where $x'_1 = 1$, $x'_2 = p(X_2)$ and $x'_3 = p(X_3)$.

    \proofsubparagraph*{Description of the algorithm.} First, we compute $\Zz = \{p(X) \mid X \subseteq A, |X| \le |A|/4\}$ and then compute $\Rr_t$ using \cref{eq:Rt-defn} for all $t \in \Zz$. For every $t \in \Zz$ do the following: using \cref{lem:lp_query}, set up a data structure $\Dd_t$ supporting linear programming queries on $3$ dimensions, preprocessed with the elements $(z'_1, z'_2, z'_3) \in \Rr_t$. Note that $\Dd_t$ is capable of answering linear programming queries of the form $\min_{(z'_1, z'_2, z'_3) \in \Rr_t} \left\{\sum_{i = 1}^{3}z'_i \cdot x'_i\right\}$ for any queried point $(x'_1, x'_2, x'_3)$.
    We now iterate over all partitions $X_1 \uplus X_2 \uplus X_3 \uplus X_4 = A$ with $|X_1| \le |X_2| \le |X_3| \le |X_4|$, and for each partition do the following. Let $t = p(X_1)$. We compute the value $\xi = \sum_{i \in [4]}\cost(X_i)$. Next, we set $(x'_1 = 1, x'_2 = p(X_2), x'_3 = p(X_3))$, and query $\Dd_t$ for the value $q = \min_{(z'_1, z'_2, z'_3) \in \Rr_t} \left\{\sum_{i=1}^{3} z'_i \cdot x'_i\right\}$. 
    Finally, we output the minimum value of $(\xi + q)$ over all partitions $X_1 \uplus X_2 \uplus X_3 \uplus X_4 = A$. We present pseudocode of the approach in \cref{alg:mitm-P4-wjCj}.

    \begin{algorithm}[!htbp]
        \caption{\PmSwjCj{4} for jobs $i \in [n]$ having weights $w_i$ and processing times $p_i$}\label{alg:mitm-P4-wjCj}
        \begin{algorithmic}[1]
            \State{Sort jobs $i \in [n]$ in a non-decreasing order of $\frac{p_i}{w_i}$.}
            \State{$\alpha \gets 0.628$, $A \gets [\floor{\alpha n}]$, $B \gets [n] \setminus A$}
            \State{$\Zz \gets \{p(X) \mid X \subseteq A, |X| \le |A|/4\}$}
            \For{$t \in \Zz$}
                \State{$\Rr_t = \left\{(z'_1, z'_2, z'_3) \mid Y_1 \uplus Y_2 \uplus Y_3 \uplus Y_4 = B \right\}$} \Comment{as defined in \cref{eq:Rt-defn}}
                \State{Use \cref{lem:lp_query} to set up a data structure $\Dd_t$ from $\Rr_t$.}
            \EndFor
            \State{$\ans \gets \infty$}
            \For{$X_1 \uplus X_2 \uplus X_3 \uplus X_4 = A$ with $|X_1| \le |X_2| \le |X_3| \le |X_4|$}
                \State{$t \gets p(X_1)$}
                \State{$\xi \gets \sum_{i \in [4]}\cost(X_i)$}
                \State{$x'_1 \gets 1$, $x'_2 \gets p(X_2)$, $x'_3 \gets p(X_3)$}
                \State{Query $\Dd_t$ for $q \gets \min_{(z'_1, z'_2, z'_3) \in \Rr_t} \left\{\sum_{i=1}^{3} z'_i \cdot x'_i\right\}$}
                \State{$\ans \gets \min\left(\ans, \xi + q\right)$}
            \EndFor
            \State{\Return $\ans$}
        \end{algorithmic}
    \end{algorithm}  

    \proofsubparagraph*{Proof of correctness.}
    Fix a partition $X_1 \uplus X_2 \uplus X_3 \uplus X_4 = A$ and set $t = p(X_1)$. 
    The discussion on the mathematical representation of the weighted completion time implies that the computed value $(\xi + q) = \sum_{i \in [4]}\cost(X_i) + \min_{(z'_1, z'_2, z'_3) \in \Rr_t} \left\{\sum_{i=1}^{3} z'_i \cdot x'_i\right\}$ is equal to the minimum total weighted completion time when $X_i$ is scheduled completely on the $i$-th machine. The algorithm outputs the minimum value of $(\xi + q)$ over all choices of $X_1 \uplus X_2 \uplus X_3 \uplus X_4 = A$, and hence correctly outputs the minimum weighted completion time for the input instance.
    
    \proofsubparagraph*{Running-time analysis.} First, observe that $|\Zz| \le \binom{|A|}{\le|A|/4} \le \Os(1.7548^{|A|})$ by \cref{prop:binom-to-entropy}, and hence $\Zz$ can be computed in time $\Os(1.7548^{|A|})$. For every $t \in \Zz$, $|\Rr_t| \le 4^{|B|}$ and $\Rr_t$ can be computed in time $\Os(4^{|B|})$. Moreover, for every $t \in \Zz$, preprocessing $\Dd_t$ takes time $\Oh\left(|\Rr_t|\log|\Rr_t|\right) = \Os(4^{|B|})$ (\cref{lem:lp_query}). The total preprocessing time across all $t \in \Zz$ is $|\Zz| \cdot \Os(4^{|B|}) = \Os(1.7548^{|A|} \cdot 4^{|B|})$. For any $t \in \Zz$, a linear programming query to $\Dd_t$ takes time $\Oh\left(\log|\Rr_t|\right) = {\Oh(n)}$ (\cref{lem:lp_query}). In total, the algorithm makes a total of $\Os(4^{|A|})$ queries to the data structures $\Dd_t$, over all $t \in \Zz$. Hence these queries take a total time of $\Oh(n) \cdot \Os(4^{|A|}) = \Os(4^{|A|})$.
    The total running time of the entire algorithm is therefore $\Os\left(1.7548^{|A|} \cdot 4^{|B|} + 4^{|A|}\right) = \Os\left(1.7548^{\alpha n} \cdot 4^{(1-\alpha)n} + 4^{\alpha n}\right)$. This is minimized when $1.7548^{\alpha n} \cdot 4^{(1-\alpha)n} = 4^{\alpha n}$, i.e.~when $\alpha \approx 0.628$. By our choice of $\alpha = 0.628$ this gives us a running time of $\Os\left(1.7548^{\alpha n} \cdot 4^{(1-\alpha)n} + 4^{\alpha n}\right) \le \Oh(2.389^n)$.
\end{proof}

To obtain algorithms solving \PmSwjCj{m} for $m \in \{5, 6\}$ in time faster than $\Oh(2.755^n)$, we crucially use the equalities in \cref{lem:dp-correct-recursion} with $j = 3$ as follows.
We use a dynamic program similar to the one in the proof of \cref{thm:SumwjCj}.
For any subset $S \subseteq [n]$ and $i \in [m]$, let the table entry $\DP[i, S]$ contains the minimum total weighted completion time when the jobs in $S$ are scheduled on $i$ machines.

\begin{restatable}{lemma}{PfivewjCj}\label{lem:P5wjCj}
   \PmSwjCj{5} can be solved in $\Oh(2.726^{n})$ time.
\end{restatable}

\begin{proof}
    For every $S \subseteq [n]$, initialize $\DP[2, S]$ using \cref{lem:mitm-wjCj}. Moreover, for every $S \subseteq [n]$ with $|S| \le 0.6n $, initialize $\DP[3, S]$ using \cref{lem:mitm-wjCj}.
    By \cref{lem:dp-correct-recursion} with $i= 5$, $j = 3$ and $S = [n]$ the following recurrence holds.
    \[\DP[5, [n]] = \min \limits_{S \in \binom{[n]}{\le 3n/5}} \{\DP[3, S] + \DP[2, [n] \setminus S]\} = \min \limits_{S \in \binom{[n]}{\le 0.6n }} \{\DP[3, S] + \DP[2, [n] \setminus S]\}.\]
    With the values of $\DP[2, S]$ for $S \subseteq [n]$ and $\DP[3, S]$ for $S \in \binom{[n]}{\le 0.6n}$ already available, compute the value of $\DP[5, [n]]$ by enumerating all $S \in \binom{n}{\le 0.6n }$ and taking the minimum value of $\DP[3, S] + \DP[2, [n] \setminus S]$. Finally,  output this computed value of $\DP[5, [n]]$. 
    
    The correctness of the algorithm follows from \cref{lem:mitm-wjCj,lem:dp-correct-recursion}. To analyze the running time, note that the initialization for $i=2$ and all $S \subseteq [n]$ takes time at most
    \[\sum \limits_{S \subseteq [n]} \Os\left(2^{|S|/2}\right) = \Os \left(\sum \limits_{r = 0}^{n} \binom{n}{r}\cdot \left(\sqrt{2}\right)^{r}\right) = \Os\left(\left(1+\sqrt{2}\right)^n\right) \le \Oh(2.4143^n).\] 
    The initialization for $i = 3$ and all $S \in \binom{[n]}{\le 0.6n}$ takes time at most 
    \[\sum \limits_{S \in \binom{[n]}{\le 0.6n }} \Os\left(3^{|S|/2}\right) = \Os \left(\sum \limits_{r = 0}^{\floor{0.6n} } \binom{n}{r} \cdot \left(\sqrt{3}\right)^{r}\right) \le \Oh(2.726^n).\]
    The last inequality is due to \cref{lem:0.6-binom}. Finally, the computation of $\DP[5, [n]]$ using the above recurrence takes time $\Os(2^n)$. 
    In total, the algorithm runs in time $\Oh(2.4143^n) + \Oh(2.726^n) + \Os(2^n) = \Oh(2.726^n)$.
\end{proof}

\begin{restatable}{lemma}{PsixwjCj}\label{lem:P6wjCj}
    \PmSwjCj{6} can be solved in $\Oh(2.733^{n})$ time.
\end{restatable}

\begin{proof}

    For every $S \subseteq [n]$, initialize $\DP[3, S]$ using \cref{lem:mitm-wjCj}.
    By \cref{lem:dp-correct-recursion} with $i = 6$, $j = 3$ and $S = [n]$ the following recurrence holds.
    \[\DP[6, [n]] = \min_{S \subseteq [n]} \{\DP[3, S] + \DP[3, [n] \setminus S]\}.\]
    With the values of $\DP[3, S]$ for $S \subseteq [n]$ already available, compute the value of $\DP[6, [n]]$ by enumerating all $S \subseteq [n]$ and taking the minimum value of $\DP[3, S] + \DP[3, [n] \setminus S]$. Finally, output this computed value of $\DP[6, [n]]$. 
    
    The correctness of the algorithm follows from \cref{lem:mitm-wjCj,lem:dp-correct-recursion}. To analyze the running time, note that the initialization for $i=3$ and all $S \subseteq [n]$ takes time at most
    \[\sum \limits_{S \subseteq [n]} \Os\left(3^{|S|/2}\right) = \Os \left(\sum \limits_{r = 0}^{n} \binom{n}{r}\cdot \left(\sqrt{3}\right)^{r}\right) = \Os\left(\left(1+\sqrt{3}\right)^n\right) \le \Oh(2.733^n).\]
    Finally, the computation of $\DP[6, [n]]$ using the above recurrence takes time $\Os(2^n)$. 
    In total, the algorithm runs in time $\Oh(2.733^n) + \Os(2^n) = \Oh(2.733^n)$.
\end{proof}

 \section{Weighted number of tardy jobs}\label{sec:wjUj}

In this section, we focus on the problem of minimizing the total weight of tardy jobs \PSwjUj defined in \cref{sec:tech_overview}. 
We will use the Fast Subset Convolution algorithm due to Bj\"{o}rklund et al.~\cite{bjorklund2007fourier}.

\begin{lemma}[Fast Subset Convolution, follows from~{\cite[Theorem 3]{bjorklund2007fourier}}]\label{lem:fast-subset-conv}
    Given functions $f,g: 2^{[n]} \to \{0\} \cup [n]$, the min-plus subset convolution $h(S) = \min_{S' \subseteq S}\{f(S') + g(S \setminus S')\}$ can be computed in time $\Os(2^n)$.
\end{lemma}

We start by recalling an algorithm for solving the unweighted version of this problem on a single machine, this problem is denoted as $\PmSUj{1}$.
\begin{lemma}[Moore's algorithm~\cite{moore1968n}]\label{lem:polyP1Uj}
    $\PmSUj{1}$ can be solved in $\Oh(n \log n)$ time.
\end{lemma}

We prove \cref{thm:SumwjUj} via standard dynamic programming, where the base case is handled using \cref{lem:polyP1Uj}, and the transitions are handled using \cref{lem:fast-subset-conv}.

\SumwjUj*

\begin{proof}
    Consider an instance of \PSwjUj, i.e.,~consider $n$ jobs with processing times $p_1, \dots, p_n$, weights $w_1, \dots, w_n$ and due dates $d_1, \dots, d_n$, and let $m$ be the number of available identical parallel machines. 
    Note that if $m \geq n$, then each job can be processed at time $0$ on a distinct machine so that it is tardy if and only if $p_j > d_j$. Thus, the optimal weighted number of tardy jobs is simply $\sum_{i=1}^n w_i \cdot \llbracket p_i > d_i \rrbracket$. 
    Hence, suppose that $m < n$.
    We construct a DP table such that for any subset $S \subseteq [n]$ and $i \in [m]$, the table entry $\DP[i, S]$ contains the minimum total number of tardy jobs when the jobs in $S$ are scheduled on $i$ machines. 
    In other words, $\DP[i, S]$ stores the optimum for the \emph{unweighted version} of the problem on $i$ machines with jobs $S$.
    We prove the following claim.
    \begin{claim}\label{clm:Uj-to-wjUj}
        The minimum total weight of tardy jobs is equal to 
        \[\min\{w([n] \setminus S) \mid S \subseteq [n], \DP[m, S] = 0\}.\]
    \end{claim}
    \begin{claimproof}
        Let the minimum total weight of tardy jobs be $\OPT$. 
        Fix an optimal schedule.
        Let $S^* \subseteq [n]$ be the set of jobs which are not tardy in this schedule. 
        Since the jobs $S^*$ can be scheduled on $m$ machines with every job meeting its deadline, we have $\DP[m, S^*] = 0$. 
        Moreover, the total weight of tardy jobs is $w([n] \setminus S^*)$. 
        Therefore, we have
        \[\OPT = w([n] \setminus S^*) \ge \min\{w([n] \setminus S) \mid S \subseteq [n], \DP[m, S] = 0\}.\]
        On the other direction, let $S^\dagger \subseteq [n]$ be a subset with $\DP[m, S^\dagger] = 0$ which minimizes $w([n] \setminus S^\dagger)$. 
        Since $\DP[m, S^\dagger] = 0$, the jobs $S^\dagger$ can be scheduled on $m$ machines with every job meeting its deadline. 
        Extending this schedule by scheduling the remaining jobs $[n] \setminus S^\dagger$ arbitrarily after the already scheduled jobs, we get the total weight of tardy jobs in this schedule to be at most $w([n] \setminus S^\dagger)$, and this must be at least $\OPT$. 
        Therefore,
        \[\OPT \le w([n] \setminus S^\dagger) = \min\{w([n] \setminus S) \mid S \subseteq [n], \DP[m, S] = 0\}.\]
        The above inequalities together imply $\OPT = \min\{w([n] \setminus S) \mid S \subseteq [n], \DP[m, S] = 0\}$.
    \end{claimproof}
    We now describe the algorithm.
    For every $S \subseteq [n]$, initialize $\DP[1, S]$ using \cref{lem:polyP1Uj}. 
    By \cref{lem:dp-correct-recursion}, the following recurrence holds for any $i \geq 2$ and $S \subseteq [n]$.
    \[\DP[i, S] = \min \limits_{S' \subseteq S} \{ \DP[1, S'] + \DP[i - 1, S \setminus S'] \}.\]
    We can therefore iterate over $i \in \{2, \dots, m\}$ in increasing order, and compute all table entries $\DP[i, S]$ using Fast Subset Convolution (\cref{lem:fast-subset-conv}) with the above recurrence. 
    In the end, we output $\min\{w([n] \setminus S) \mid S \subseteq [n], \DP[m, S] = 0\}$.

    The correctness of the algorithm follows from \cref{clm:Uj-to-wjUj,lem:dp-correct-recursion,lem:fast-subset-conv}. 
    To analyze the running time, note that the initialization for $i=1$ takes time at most $\Oh(2^n) \cdot \Oh(n \log n) = \Os(2^n)$. For every $i \in \{2, \ldots, m\}$, Fast Subset Convolution using \cref{lem:fast-subset-conv} takes time $\Os(2^n)$, as all entries in $\DP$ are at most $n$. Finally, computing $\min\{w([n] \setminus S) \mid S \subseteq [n], \DP[m, S] = 0\}$ takes a total time of $\Os(2^n)$. The total running time adds up to $\Os(2^n + m \cdot 2^n + 2^n) = \Os(2^n)$. 
\end{proof} \section{Bin Packing}\label{sec:binpacking}

In this section, we focus on the \binpacking problem and show how to beat the classic $\Os(2^n)$-time algorithm when assuming ARC.

\BPARC*

We only focus on solving the decision version of \binpacking (which we denote as \binpacking) as defined in \cref{sec:intro}.
Indeed, we can reconstruct a solution with $n^{\Oh(1)}$ calls to a decision oracle as follows.
Consider a \binpacking instance $(m, t, I)$ with $n \coloneq |I|$ that admits a solution.
If $m = n$ and $x \le t$ for every $x \in I$, then we can return the solution that packs each item in a separate bin.
We cannot have $x > t$ for some $x \in I$, as this would contradict that there is a solution.
Otherwise, guess a pair of items $x, y \in I$ and replace them by the single item of value $x+y$. If an oracle for \binpacking called on the new instance answers \textsc{Yes}, then we know that there is a feasible solution where $x$ and $y$ are in the same bin, and so we can continue the procedure on the smaller instance.
Since the given instance admits a solution, there exists such a pair for which the decision oracle yields \textsc{Yes} on the smaller instance. In each step, we require at most $\Oh(n^2)$ oracle calls to find that pair, and since the instance size decreases by 1 in each step, there are at most $\Oh(n)$ steps. Hence, in total we require $\Oh(n^3)$ oracle calls to construct the solution.

Recall that when the number of bins is a fixed constant, then \binpacking can be solved in time $\Oh((2-\eps_m)^n)$ as shown by Nederlof et al.~\cite{nederlof2023faster}, where $\eps_m > 0$ is a constant that depends on the number of bins $m$.

\lemmaBinPackingMFixed*{}

In \cref{thm:bin-packing-arc}, the number of bins is not fixed, however we will show how to exploit the algorithm of Nederlof et al.~\cite{nederlof2023faster} on smaller instances.
We will also reduce to special instances of \threepartition, which can be solved in time $\Oh((2-\eps)^n)$ for some $\eps > 0$ when assuming ARC as proven by Björklund et al.~\cite{radu25}.
We restate the corresponding lemma here for convenience.

\lemmaFastThreePartitionARC*{}

The idea to prove \cref{thm:bin-packing-arc} is to differentiate between four different types of \binpacking instances, and to provide an efficient algorithm for each case.
We now describe the different instance types.

Consider a \binpacking instance $(m, t, I)$.
Note that without loss of generality, we can assume that $I \subset [t]$, as otherwise we can directly answer \textsc{No}.
Then, if $m \geq n$, we can directly answer \textsc{Yes} as each item can be put in its own bin, therefore we can also assume $m < n$.
Any partition of $I$ into $m$ parts $X_1, \dots, X_m$ is called a \emph{packing} of $I$ into $m$ bins. The quantity $\Sigma(X_i)$ is called the \emph{load of the $i$-th bin}. The packing forms a \emph{solution} if $\Sigma(X_i) \leq t$ for all $i \in [m]$, i.e.~the load of each bin doesn't exceed the capacity $t$.
The \emph{cardinalities of a solution} $X_1, \dots, X_m$ are the integers $b_1, \dots, b_m$ such that $b_i \coloneq |X_i|$ for all $i \in [m]$.
Note that we can assume without loss of generality that $b_1 \geq \dots \geq b_m$.

Then, we distinguished between different \binpacking instance as follows. Assume that $(m, t, I)$ admits a solution of cardinalities $b_1 \geq \dots \geq b_m$. Let $\eps >0$ be a rational constant such that $\eps < 1/100$ and $H(\eps) < \eps_6$, where $H(\cdot)$ is the entropy function defined in \cref{sec:preliminaries}.

\begin{caseenum}
    \item If $b_1 \geq (1+\eps)\cdot n / 2$, i.e.~\emph{more than half of the items are packed into the first bin}, then we can guess $\leq (1-\eps)n/2$ items and use standard dynamic programming to verify that they can be packed into $m-1$ bins.
    See \cref{lem:case_DP} for details.
    \label{case:DP}

    \item If $\sum_{i=1}^6 b_i \geq (1-\eps)\cdot n$, i.e.~\emph{almost all items are packed into the first 6 bins},
    then we can guess $\leq \eps n$ items and use standard dynamic programming to verify that they can be packed into $m-6$ bins, and then use \cref{thm:nederlof_binpacking} to verify that the $\geq (1-\eps)n$ items can be packed into 6 bins.
    See \cref{lem:case_cstBP} for details.
    \label{case:cstBP}

    \item If $b_1 \in [\frac{n}{2} \pm \eps\cdot \frac{n}{2}]$ and $\sum_{i=1}^6 b_i < (1-\eps) \cdot n$, i.e.~\emph{the first bin contains approximately half of the items, but the remaining items are not concentrated on the first 6 bins}, then, in \cref{lem:case_balanced},
    we use a randomized strategy that samples witnesses of the existence of such a packing.
    This algorithm is inspired by the proof of \cite[Lemma~3.5]{nederlof2023faster}.
    \label{case:balanced}

    \item If $b_1 \leq (1-\eps)\cdot n / 2$ and $\sum_{i=1}^6 b_i < (1-\eps) \cdot n$, i.e.~\emph{the first bin contains less than half of the items, and the items are not concentrated on the first 6 bins}, we reduce solving the problem to solving specific instances of \threepartition{}, which we then solve using \cref{thm:arc_3partition}. This is the only case in which we make use of the ARC. Refer to \cref{lem:case_ARC} for the details.
    \label{case:ARC}
\end{caseenum}
For each of the cases $x \in \{\text{A, B, C, D}\}$, we will prove that there exists a faster than $2^n$-time algorithm that, given a \binpacking instance, outputs \textsc{No} if the instance has no solution and outputs \textsc{Yes} (with high probability) if the instance admits a solution \emph{in case x}.
If the instance admits a solution falling in one of the other cases, we don't give any guarantee on the output of the algorithm. To solve any \binpacking instance, we simply run the algorithms for each of the four cases and output \textsc{Yes} if and only if one of the algorithms outputs \textsc{Yes}.

Before showing how to tackle each of the above cases, we describe a folklore dynamic-programming (DP) table which will be useful in each of the cases.
Fix a \binpacking instance $(m, t, I)$ and define the following DP table entry for every $i \in \N$ and $S \subseteq I$.
\begin{align}\label{eq:DP_def}
    \begin{split}
        \BP[i, S] \coloneq \min \{\Sigma (D) \ |\
        S = X_1 \uplus \dots \uplus X_i \uplus D \text{ such that } & \Sigma(D) \leq t                                   \\ \text{ and }
                                                                    & \Sigma(X_j) \leq t \ \text{ for all } j \in [i] \}
    \end{split}
\end{align}

$\BP[i, S]$ has value $0$ if and only if $S$ can be packed into $i$ bins of capacity $t$, i.e.~the \binpacking instance $(i, t, S)$ has a solution. Additionally, if it cannot be packed into $i$ bins, but it can be packed into $i+1$ bins, then the table entry contains the minimum load among the $i+1$ bins.
Moreover, if it is not possible to pack $S$ into $i+1$ bins, then by our convention that the minimum of an empty set is $\infty$ the table entry $\DP[i,S] = \infty$.
We show below how to compute DP table entries for a collection of sets $S$.

\begin{lemma}\label{lem:binpacking_DP_computation}
    For any $m \in \N$ and collection of multisets $\Ww$ from some universe of size $n$ of natural numbers, the table entries $\BP[i, S]$ for all $S \in \Ww$ and $i \in [m]$ can be computed in time $\Oh(n \cdot m \cdot |\dc \Ww|)$.
\end{lemma}
\begin{proof}
    First, initialize $\BP[i, \emptyset] \gets 0$ for all $0 \leq i \leq m$.
    Moreover, $\BP[0, X]$ is initialized to $\Sigma(X)$ for every $X \in \dc \Ww$ with $\Sigma(X) \leq t$, and $\BP[0, X]$ is initialized to $\infty$ for every $X \in \dc \Ww$ with $\Sigma(X) > t$.

    Next, iterate over every $i \in [m]$ in increasing order and over every $X \in \dc \Ww$ in increasing order of size, and set the DP entry $\BP[i, X]$ to
    \begin{equation}\label{eq:binpacking_DP}
        \min_{a \in X}
        \left( \min
        \left\{
        \begin{aligned}
                 & a + \BP[i, X \setminus \{a\}] + \infty \cdot \llbracket \BP[i, X \setminus \{a\}] > t - a \rrbracket, \\
                 & \infty \cdot \llbracket \BP[i-1, X \setminus \{a\}] > t - a \rrbracket
            \end{aligned}
        \right\} \right).
    \end{equation}
    Note that the DP entries in \cref{eq:binpacking_DP} are already computed since we iterate over $i$ and $X$ in increasing order.
    Each transition takes time $\Oh(|X|) = \Oh(n)$ and there are $\Oh(m \cdot |\dc \Ww|)$ entries to compute, so the total running time is at most $\Oh(n \cdot m \cdot |\dc \Ww|)$.

    It remains to show correctness of our algorithm.
    First, observe that for $i = 0$ and from our initialization we indeed get that $\BP[0,X]$ evaluates to \cref{eq:DP_def} for all $X \in \dc \Ww$.

    Now, we show by induction on $|X|$ that for every $1 \leq i \leq m$ and $X \in \dc \Ww$, the constructed DP entry $\BP[i, X]$ correctly contains the value given by \cref{eq:DP_def}.
    Since the empty set can always be packed, \cref{eq:DP_def} yields $0$ for $X = \emptyset$ and any $0 \leq i \leq m$.
    Hence, in the base case $X = \emptyset$, the entries $\BP[i, X]$ are correctly initialized to 0 for any $0 \leq i \leq m$.
    Now consider a non-empty set $X \in \dc \Ww$ and assume that for any $1 \leq i \leq m$ and $X' \in \dc \Ww$ with $|X'| < |X|$ the constructed DP entry $\BP[i, X']$ contains the correct value given by \cref{eq:DP_def}, and also recall that $\BP[0,X']$ also contains the correct value by our previous elaborations.

    Fix $1 \leq i \leq m$.
    Let $d$ be the value given by \cref{eq:DP_def} for $i$ and $X$, and let $\tilde d$ be the value computed in \cref{eq:binpacking_DP} for $i$ and $X$. We argue that $\tilde d = d$.

    We first show that $\tilde d \leq d$.
    Observe that $d \in \{0, \dots, t\} \cup \{\infty\}$.
    If $d = \infty$, then we directly get $\tilde d \leq d$.
    So suppose that $d \in \{0, \dots t\}$. Then there exists a packing $X_1, \dots, X_i, D$ of $X$ into $i+1$ bins such that $d = \Sigma(D) \leq \Sigma(X_j) \leq t$ for every $j \in [i]$.

    Suppose that $d = 0$.
    Then, $D = \emptyset$ and there exists $a \in X_j$ for some $X_j$.
    Observe that $X_1, \dots, X_{j-1}, X_{j+1},\dots, X_i, D'$ with $D' \coloneq X_j \setminus \{a\}$ is a partition of $X \setminus \{a\}$ into $i$ bins.
    In particular, $\BP[i-1, X \setminus \{a\}] \leq \Sigma(D') = \Sigma(X_j) - a \leq t - a$. This means that the second term in \cref{eq:binpacking_DP} is
    \begin{align*}
        \infty \cdot \llbracket \BP[i-1, X \setminus \{a\}] > t - a \rrbracket = 0 = d
    \end{align*}
    and thus $\tilde d \leq d$.

    Now suppose that $d \in [t]$.
    Then, $D \neq \emptyset$ and there exists some $a \in D$.
    Observe that $X_1, \dots, X_i, D'$ with $D' \coloneq D \setminus \{a\}$ is a partition of $X \setminus \{a\}$ into $i+1$ bins.
    In particular, $\BP[i, X \setminus \{a\}] \leq \Sigma(D') = \Sigma(D) - a \leq t - a$.
    This means that the first term in \cref{eq:binpacking_DP} is
    \begin{align*}
        a + \BP[i, X \setminus \{a\}] + \infty \cdot \llbracket \BP[i, X \setminus \{a\}] > t - a \rrbracket
         & = a + \BP[i, X \setminus \{a\}] \\
         & \leq a + \Sigma(D')             \\
         & = a + \Sigma(D) - a             \\
         & = \Sigma(D) = d
    \end{align*}
    and thus $\tilde d \leq d$.

    We now show that $d \leq \tilde d$.
    By \cref{eq:binpacking_DP}, since $X\neq \emptyset$, there exists $a \in X$ such that either (1) $\tilde d = a + \BP[i, X \setminus \{a\}] + \infty \cdot \llbracket \BP[i, X \setminus \{a\}] > t - a \rrbracket $ or (2) $\tilde d = \infty \cdot \llbracket \BP[i-1, X \setminus \{a\}] > t - a \rrbracket$.

    If $\tilde d = \infty$, then $d \leq \tilde d$ holds.
    So, assume that $\tilde d$ is not $\infty$.

    If we are in case (1), then necessarily $\BP[i, X \setminus \{a\}] \leq t - a$ and $\tilde d = a + \BP[i, X \setminus \{a\}]$.
    Then, there exists a partition $X_1, \dots, X_i, D$ of $X \setminus \{a\}$ such that $\Sigma(D) \leq \Sigma(X_j) \leq t$ for every $j \in [i]$ and $\Sigma(D) = \BP[i, X \setminus \{a\}] \leq t - a$.
    Therefore, $X_1, \dots, X_i, D'$ with $D' = D \cup \{a\}$ is a partition of $X$ into $i +1$ parts such that $\Sigma(X_j) \leq t$ for every $j \in [i]$ and $\Sigma(D') = \Sigma(D) + a \leq t - a + a = t$.
    By \cref{eq:DP_def}, we thus have $d \leq \Sigma(D') = \tilde d$.

    On the other hand, if we are in case (2), then $\tilde d =0$ and $\BP[i-1, X \setminus \{a\}] \leq  t - a $.
    Then, there exists a partition $X_1, \dots, X_{i-1}, D$ of $X \setminus \{a\}$ such that $\Sigma(D) \leq \Sigma(X_j) \leq t$ for every $j \in [i-1]$ and $\Sigma(D) = \BP[i-1, X \setminus \{a\}]$.
    Since $\Sigma(D) + a \leq t-a + a =t$, this means that $X_1, \dots, X_{i-1}, X_i, D'$ with $X_i = D \cup \{a\}$ and $D' = \emptyset$ is a partition of $X$ into $i+1$ bins such that $\Sigma(X_j) \leq t$ for all $j \in [i]$ and $\Sigma(D') = 0$. By \cref{eq:DP_def} we deduce that $d \leq 0 = \tilde d$.
\end{proof}

We now show how to handle each instance type of \binpacking, starting with Case~\ref{case:DP}.

\begin{lemma}[Case~\ref{case:DP}]\label{lem:case_DP}
    There exists a deterministic algorithm running in time $\Os(2^{(1-\delta)n})$ for $\delta \coloneq 1- H(\frac{1-\eps}{2})$ that, given an instance of \binpacking, outputs:
    \begin{itemize}
        \item \textnormal{\textsc{No}} if the given \binpacking instance has no solution;
        \item \textnormal{\textsc{Yes}} if the given instance of \binpacking has a solution with cardinalities $b_1 \geq \dots \geq b_m$ such that $b_1 \geq (1+\eps)\cdot n / 2$.
    \end{itemize}
\end{lemma}
\begin{proof}
    First, note that since $0< \eps <1$, we have $0 < \frac{1-\eps}{2} <1/2$, so $0 < H(\frac{1-\eps}{2}) <1$ and thus $\delta \in (0, 1)$.

    Consider a \binpacking instance $(m, t, I)$ and denote by $\Ww$ the collection of subsets of $I$ of size at most $(1-\eps) \cdot n/2$. For every $X \in \Ww$, compute $\BP[m-1, X]$ and return \textsc{Yes} if and only if there is a subset $X \in \Ww$ such that $\BP[m-1, X] = 0$ and $\Sigma(I) - \Sigma(X) \leq t$.
    By \cref{lem:binpacking_DP_computation} this can be done in total time $\Oh(n \cdot (m-1) \cdot |\dc \Ww|)$.
    Since $|\dc \Ww|  = |\Ww| \leq \Oh\left(\binom{n}{(1-\eps)n/2}\right)$, the total running time is at most $\Os(\binom{n}{(1-\eps)n/2}) \leq \Os(2^{n H(\frac{1-\eps}{2})}) = \Os(2^{(1-\delta)n})$.

    By assumption, there exists a solution $X_1, \dots, X_m$ with $|X_1| \geq (1+\eps) \cdot n/2$. Let $X \coloneq I \setminus X_1$. Then $|X| \leq (1-\eps)n/2$, $\Sigma(I) - \Sigma(X) = \Sigma(X_1)  \leq t$ and $X$ can be packed into $m-1$ bins, i.e.~$\BP[m-1, X] = 0$. Hence, the above algorithm returns \textsc{Yes}.
    Conversely, assume that the algorithm returns \textsc{Yes}. Then there exists $X \subset I$ of size at most $(1-\eps) \cdot n/2$ that can be packed into $m-1$ bins and such that $\Sigma(I \setminus X) \leq t$. Let $X_1 \coloneq I \setminus X$ and $X_2, \dots, X_m$ be the content of the $m-1$ bins packing $X$. Then, $X_1, \dots, X_m$  partitions $I$ with $|X_1| \geq (1 + \eps)n/2$ and $\Sigma(X_i) \leq t$ for every $i \in [m]$.
\end{proof}

We proceed to Case~\ref{case:cstBP}.

\begin{lemma}[Case \ref{case:cstBP}]\label{lem:case_cstBP}
    There exists an algorithm running in time $\Os(2^{(1-\delta)n})$ for $\delta \coloneq \eps_6 - H(\eps) \in (0, 1)$ that, given an instance of \binpacking with $m\geq 6$ bins, outputs:
    \begin{itemize}
        \item \textnormal{\textsc{No}} if the given instance of \binpacking has no solution;
        \item \textnormal{\textsc{Yes}} with high probability if the given instance of \binpacking has a solution with cardinalities $b_1 \geq \dots \geq b_m$ such that $\sum_{i=1}^6 b_i \geq (1-\eps)\cdot n$.
    \end{itemize}
\end{lemma}
\begin{proof}
    Consider a \binpacking instance $(m, t, I)$ and denote by $\Ww$ the collection of subsets of $I$ of size at most $\eps \cdot n$. For every $X \in \Ww$, compute the DP entry $\BP[m-6, X]$.
    This can be done in total time $\Oh(n \cdot (m-6) \cdot |\dc \Ww|)$ by \cref{lem:binpacking_DP_computation}.
    Then iterate over every $X \in \Ww$ and decide using \cref{thm:nederlof_binpacking} whether $I \setminus X$ can be packed into 6 bins.
    This takes in total time $\Os(|\Ww| \cdot 2^{(1-\eps_6)n})$.
    Return \textsc{Yes} if and only if for some $X \in \Ww$ we decided that $I \setminus X$ can be packed in 6 bins and $\BP[m-6, X] = 0$.
    Note that $|\dc \Ww| = |\Ww| \leq \Oh(\binom{n}{\eps \cdot n}) \leq \Oh(2^{n H(\eps)})$. So the total running time is at most
    \begin{align*}
        \Oh(n \cdot (m-6) \cdot |\dc \Ww|) + \Os(|\Ww| \cdot 2^{(1-\eps_6)n})
         & = \Os(2^{n H(\eps)} \cdot 2^{(1-\eps_6)n}) = \Os(2^{(1-\delta)n})
    \end{align*}
    for $\delta= \eps_6 - H(\eps)$. Since $0 < H(\eps) < \eps_6 < 1$ by our choice of $\eps$, we correctly have $\delta \in (0, 1)$.

    It remains to argue correctness. Recall that $\BP[i, X] = 0$ if and only if $X$ can be packed into $i$ bins.
    Hence, if the above procedure returns \textsc{Yes}, then there exists $X \subset I$ of size at most $\eps \cdot n$ such that $I \setminus X$ can be packed into $6$ bins and $X$ can be packed into $m-6$ bins.
    Let $X_1, \dots, X_6$ be the content of the $6$ bins packing $I \setminus X$ and let $X_7, \dots, X_m$ be the content of the $m-6$ bins packing $X$.
    Then $X_1, \dots, X_m$ is a partition of $I$ such that $\Sigma(X_i) \leq t$ for every $i \in [m]$.
    Additionally, we have $|I \setminus X| = \sum_{i=1}^6 |X_i| \geq (1-\eps)\cdot n$.

    Conversely, assume that there exists a solution $X_1, \dots, X_m$ with cardinalities $b_1 \geq \dots \geq b_m$ such that $\sum_{i=1}^6 b_i \geq (1- \eps) \cdot n$.
    Let $X \coloneq \bigcup_{i=7}^m X_i$.
    Since $I \setminus X$ contains the items packed into the first 6 bins of the solution and $X$ contains the items packed in the $m-6$ last bins in the solution, $X$ can be packed into $m-6$ bins and $I \setminus X$ can be packed into $6$ bins.
    Additionally, we have $|X| = n - \sum_{i=1}^6 b_i \leq \eps \cdot n$. So $X \in \Ww$ and thus the above algorithm will return \textsc{Yes} with high probability.
\end{proof}

Now, we cover the approach for Case~\ref{case:balanced}.

\begin{lemma}[Case \ref{case:balanced}]\label{lem:case_balanced}
    There exists an algorithm running in time $\Os(2^{(1-\delta)n})$ for $\delta= \frac{2}{\ln 2} \left(\frac{\eps}{4 \log(12/\eps)}\right)^2$ that, given an instance of \binpacking with $m\geq 7$ bins and $n$ items where $n$ is even and $\eps n/2$ is an integer, outputs
    \begin{itemize}
        \item \textnormal{\textsc{No}} if the given instance of \binpacking has no solution;
        \item \textnormal{\textsc{Yes}} with high probability if the given instance of \binpacking has a solution with cardinalities $b_1 \geq \dots \geq b_m$ such that $b_1 \in [\frac{n}{2} \pm \eps\cdot \frac{n}{2}]$ and $\sum_{i=1}^6 b_i < (1-\eps) \cdot n$.
    \end{itemize}
\end{lemma}
\begin{proof}
    The algorithm is inspired by the proof of \cite[Lemma 3.5]{nederlof2023faster}. We show how to achieve probability at least 1/2, which can be boosted to probability at least $1- n^{-\Omega(1)}$ with a polynomial overhead.
    Consider a \binpacking instance $(m, t, I)$ with $m \geq 7$ and $|I|$ even.
    Let $\Ww$ be a collection of $2^{(1-\eps)n} n^{2}$ subsets of $I$ of size $n/2$ sampled independently and uniformly at random, and let $\Ww' = \{I \setminus X \mid X \in \Ww\}$.
    For each $X \in \Ww \cup \Ww'$ and $\ell \in [m-1]$, compute the DP entries $\BP[\ell, X]$ and $\BP[m-\ell-1, I \setminus X]$.
    Return \textsc{Yes} if and only if $\BP[\ell, X] + \BP[m-\ell-1, I \setminus X] \leq t$ for some $X \in \Ww$.
    By \cref{lem:binpacking_DP_computation}, computing all DP entries takes in total time $\Oh(n \cdot m \cdot (|\dc \Ww| + |\dc \Ww'|))$.

    We now show that $|\dc \Ww'| \leq |\uc \Ww|$ by providing an injective function $f: \dc \Ww' \to \uc \Ww$.
    For every $Y \in \dc \Ww'$, set $f(Y) \coloneq I \setminus Y$.
    Since $Y \in \dc \Ww'$, there exists $X' \in \Ww'$ such that $Y \subseteq X'$.
    By the definition of $\Ww'$, there exists $X \in \Ww$ such that $X' = I \setminus X$.
    Then, $X = I \setminus X' \subseteq I \setminus Y = f(Y)$, and hence $f(Y) \in \uc \Ww$.
    Moreover, if $f(Y) = f(Z)$ for some $Y, Z \in \dc \Ww'$, then $I \setminus Y = I \setminus Z$ and therefore $Y = Z$.
    Thus, $f$ is injective and $|\dc \Ww'| \leq |\uc \Ww|$.

    By \cref{lem:nederlof2023_bound} we have
    \begin{align*}
        |\dc \Ww| + |\uc \Ww| \leq \Oh(2^{(1-\delta)n})
    \end{align*}
    for $\delta= \frac{2}{\ln 2} \left(\frac{\eps}{4 \log(12/\eps)}\right)^2$. Thus, the algorithm takes total running time $\Os(2^{(1-\delta)n})$.

    Assume that the algorithm returns \textsc{Yes}, i.e.~there is a subset $X \in \Ww$ and $\ell \in [m-1]$ such that $\BP[\ell, X] + \BP[m-\ell-1, I \setminus X] \leq t$.
    Let $X_1, \dots, X_\ell, D_1$ be the partition of $X$ corresponding to $\BP[\ell, X]$, i.e.~$\Sigma(X_i) \leq t$ for all $i \in [\ell]$ and $\Sigma(D_1) = \BP[\ell, X] \leq t$.
    Similarly, let $X_{\ell+2}, \dots, X_m, D_2$ be the partition of $I \setminus X$ corresponding to $\BP[m-\ell-1, I \setminus X]$, i.e.~$\Sigma(X_i) \leq t$ for all $i \in \{\ell+2, \dots, m\}$ and $\Sigma(D_2) = \BP[m-\ell-1, I \setminus X] \leq t$.
    Set $X_{\ell + 1} \coloneq D_1 \cup D_2$. Then, $X_1, \dots, X_m$ is a partition of $I$ such that $\Sigma(X_i) \leq t$ for every $i \in [m]$, i.e.~$X_1, \dots, X_m$ is a solution of $(m, t, I)$.

    For the converse, assume that there exists a solution $X_1, \dots, X_m$ with cardinalities $b_1 \geq \dots \geq b_m$ such that $b_1 \in [\frac{n}{2} \pm \eps\cdot \frac{n}{2}]$ and $\sum_{i=1}^6 b_i < (1-\eps) \cdot n$.
    We show how to partition $I \setminus X_1$ into two subsets $L$ and $R$ of sizes at most $(1-\eps)n/2$.
    Initialize $L \gets X_2$ and $R \gets \bigcup_{i=3}^6 X_i$ (see \cref{fig:LR}). Then, for every $j \in \{7, \dots, m\}$ assign $X_j$ to the smallest subset of $L$ or $R$, i.e.~let $M \gets M \cup X_j$ where $M \in \argmin\{|L|, |R|\}$.
    By assumption, we have $\sum_{i=1}^6 b_i - b_1 \leq (1-\eps)n - (1-\eps)n/2 = (1-\eps)n/2$, so, at initialization
    \begin{equation}\label{eq:bp_balanced_invariant}
        |L| = b_2 \leq \sum_{i=1}^6 b_i - b_1 \leq  (1-\eps)n/2 \quad \text{ and } \quad
        |R| = \sum_{i=3}^6 b_i \leq \sum_{i=1}^6 b_i - b_1 \leq  (1-\eps)n/2.
    \end{equation}
    We argue that this holds before and after every step.  Indeed, at the beginning of step $j \in \{7, \dots, m\}$ we have $X_1 \uplus L \uplus R = X_1 \uplus \dots \uplus X_{j-1} \subset I$ and so the smallest subset in $\{X_1, L, R\}$ has size at most $\frac{1}{3}|X_1 \uplus \dots \uplus X_{j-1}| \leq n/3$. Since $|X_1| = b_1 \geq (1-\eps)n/2 >n/3$, the smallest subset is either $L$ or $R$. Furthermore, notice that $b_m \leq \dots \leq b_7 \leq n/7$. So, after step $j$, if $X_j$ was assigned to $M \in \{L, R\}$, then $|M|$ is increased to at most $n/3 + n/7 \leq (1-\eps)n/2$ (since $\eps < \frac{1}{21}$).
    So \cref{eq:bp_balanced_invariant} is indeed preserved throughout the steps, and we obtain a partition of $I \setminus X_1$ into $L$ and $R$ such that each part has size at most $(1-\eps)n/2$.

    \begin{figure}
        \centering
        \includegraphics[height=0.25\textheight]{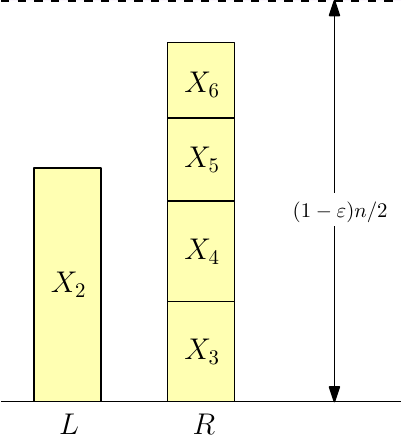}
        \caption{Initializing $L \gets X_2$ and $R \gets \bigcup_{i = 3}^{6}X_i$ in \cref{lem:case_balanced}.}\label{fig:LR}
    \end{figure}

    Without loss of generality, assume that $|L| \geq |R|$.
    Then, since $b_1 \leq (1+\eps)n/2$ we have
    \begin{align*}
        |L| \geq \frac{1}{2}|I \setminus X_1|  \geq \frac{1}{2}(n - (1+\eps)n/2)  = (1-\eps)n/4.
    \end{align*}
    This yields
    \[n/2 - |L| \leq n/2 - (1-\varepsilon)n/4 = (1+\eps)n/4 \leq (1-\eps)n/2 \leq |X_1|.\]
    Moreover, from $|L| \leq (1-\eps)n/2$ we obtain $n/2 - |L| \geq \eps n/2 \geq 0$.

    So there exists at least one subset $D \subset X_1$ of size $|D| = \frac{n}{2} - |L|$.
    For any such subset $D$, let $D' \coloneq X_1 \setminus D$ and $X_D \coloneq L \cup D$ (see \cref{fig:witness}).
    Then $X_D$ has size $n/2$, $L \subset X_D$ and $R \cap X_D = \emptyset$.
    Furthermore, if $\ell \in [m-1]$ is the number of parts $X_i$ contained in $L$, then we have $\BP[\ell, X_D] + \BP[m-\ell-1, I \setminus X_D] \leq \Sigma(D) + \Sigma(D') = \Sigma(X_1) \leq t$.
    Hence, if the algorithm samples such a subset $X_D$, then it will output \textsc{Yes}.

    It remains to analyze the probability of sampling a suitable set $X_D$.
    The number of subsets $X_D$ is at least
    \[\binom{b_1}{n/2 - |L|} \geq \binom{(1-\eps)n/2}{n/2 -|L|},\]
    since $b_1 \geq (1-\varepsilon)n/2 \geq n/2 - |L|$.

    We have already established that $n/2 - |L| \geq \eps n /2$.
    Now, observe that
    \[(1-\eps)n/2 - (n/2 - |L|) = |L| - \eps n /2 \geq (1-\eps)n/4 - \eps n /2 \geq \eps n/2\]
    since $\eps \leq 1/5$.
    By the unimodality of the binomial coefficient, we now obtain
    \[\binom{(1-\eps)n/2}{n/2 -|L|} \geq \binom{(1-\eps)n/2}{\eps n/2} \geq \binom{\eps n}{\eps n/2}.\]
    We can apply an averaging argument to further obtain
    \[\binom{\eps n}{\eps n/2} \geq \frac{1}{\eps n + 1} \sum _{r = 0}^{\eps n} \binom{\eps n}{r} = \frac{2^{\eps n}}{\eps n + 1}.\]

    Since there are $\binom{n}{n/2}$ subsets to sample from, the probability that a random subset $X_{\$}$ of size $n/2$ is of the form $X_D$ is at least
    \begin{align*}
        \Pr[\exists D \subset X_1 \ : \   X_{\$} = L \cup D] \geq \frac{\left(\frac{2^{\eps n}}{\eps n + 1}\right)}{\binom{n}{n/2}} \geq \frac{\left(\frac{2^{\eps n}}{\eps n + 1}\right)}{2^n} = 2^{(\eps - 1)n} \frac{1}{\eps n+1} \ge \frac{1}{2^{(1-\eps)n}n^2}.
    \end{align*}
    Since $|\Ww| \geq 2^{(1-\eps)n} n^{2}$, we deduce that $\Ww$ contains a suitable subset $X_D$ with probability at least $1/2$. By the above discussion, this means that with probability at least $1/2$ the algorithm outputs \textsc{Yes}.
    \begin{figure}
        \centering
        \includegraphics[width=0.95\linewidth]{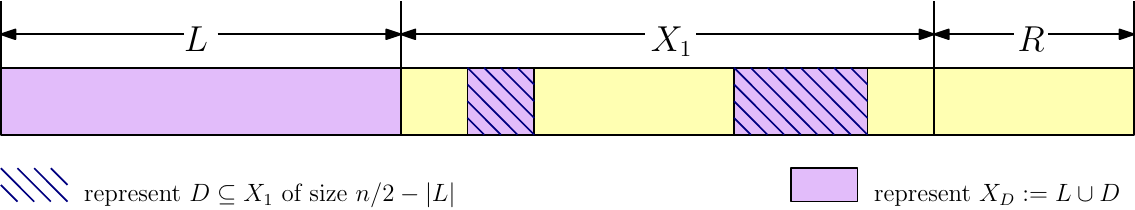}
        \caption{The sets $L$, $R$, $X_1$, $D$ and $X_D$ in \cref{lem:case_balanced}.}\label{fig:witness}
    \end{figure}
\end{proof}
Finally, we give the approach for Case~\ref{case:ARC}.
Recall the constant $\eta_\eps > 0$ of \cref{thm:arc_3partition}.

\begin{lemma}[Case \ref{case:ARC}]\label{lem:case_ARC}
    Assuming ARC, there exists a deterministic algorithm running in time $\Os(2^{(1-\delta)n})$ for $\delta =\min \{1- H(\frac{1-\eps}{2}), \eta_\eps\} \in (0, 1)$ that, given an instance of \binpacking  with $m\geq 7$ bins, outputs:
    \begin{itemize}
        \item \textnormal{\textsc{No}} if the given instance of \binpacking has no solution;
        \item \textnormal{\textsc{Yes}} if the given instance of \binpacking has a solution with cardinalities $b_1 \geq \dots \geq b_m$ such that $b_1 \leq (1-\eps)\cdot n / 2$ and $\sum_{i=1}^6 b_i < (1-\eps) \cdot n$.
    \end{itemize}
\end{lemma}
\begin{proof}
    The idea is to reduce the \binpacking instance to  \threepartition instances, which we then solve using \cref{thm:arc_3partition}.
    Let us assume that $m \leq n$, as otherwise the problem is trivial.
    Consider a \binpacking instance $(m, t, I = \{i_1,\dots,i_n\})$ and denote by $\Ww$ the collection of subsets of $[n]$ of size at most $(1-\eps)\cdot n/2$.
    For a subset $X$ of $[n]$, let $I(X)$ denote the set of items with indices in $X$.
    For every $X \in \Ww$ and $\ell \in [m]$, compute the DP entry $\BP[\ell, I(X)]$. By \cref{lem:binpacking_DP_computation}, this takes in total time $\Oh(n \cdot m \cdot |\dc \Ww|) = \Os(2^{H(\frac{1-\eps}{2})n})$.
    Next, for every $\ell_1, \ell_2, \ell_3 \in [m]$ such that $\ell_1 + \ell_2 + \ell_3 = m$, construct the families $\Ff_1, \Ff_2, \Ff_3$ defined as
    \begin{align*}
        \Ff_i \coloneq \{X \in \Ww \ : \ \BP[\ell_i, I(X)] = 0\}
    \end{align*}
    for $i \in \{1, 2, 3\}$.
    Note that $\Ff_1, \Ff_2$ and $\Ff_3$ are set families over the universe $[n]$, and the subsets in $\Ff_1, \Ff_2$ and $\Ff_3$ have size at most $(1-\eps)n/2$. So $(\Ff_1, \Ff_2, \Ff_3)$ is an instance of \threepartition over an $n$-element universe for $(1-\eps)/2$-bounded set families, which can be solved in time $\Os(2^{(1-\eta_\eps) n})$ by \cref{thm:arc_3partition}.
    Return \textsc{Yes} if and only if a solution to \threepartition was found for $(\Ff_1, \Ff_2, \Ff_3)$ for some  choice of $\ell_1, \ell_2, \ell_3 \in [m]$ with $\ell_1 + \ell_2 + \ell_3 = m$.
    In total, since there are $\Oh(m^2)$
    choices of $\ell_1, \ell_2, \ell_3$ with $\ell_1 + \ell_2 + \ell_3 = m$, the running time is at most
    \begin{align*}
        \Os(2^{H(\frac{1-\eps}{2})n} + 2^{(1-\eta_\eps)n}) = \Os(2^{(1-\delta)n})
    \end{align*}
    for $\delta \coloneq \min \{1-H(\frac{1-\eps}{2}), \eta_\eps\}$. Note that $0 < \frac{1-\eps}{2} < 1/2$ so $0 < H(\frac{1-\eps}{2}) < 1$ and thus $\delta \in (0, 1)$ as required.

    We argue correctness. Assume the algorithm outputs \textsc{Yes}, i.e.~we found $S_1 \in \Ff_1$, $S_2 \in \Ff_2$ and $S_3 \in \Ff_3$ such that $S_1 \uplus S_2 \uplus S_3 = [n]$ for some choice of $\ell_1, \ell_2, \ell_3$. Then by definition of the families, for $j \in \{1, 2, 3\}$, we have $\BP[\ell_j, I(S_j)]=0$. So let $X_1, \dots, X_{\ell_1}$ be the partition of $I(S_1)$ corresponding to $\BP[\ell_1, I(S_1)]$, $X_{\ell_1 + 1}, \dots, X_{\ell_1 + \ell_2}$ be the partition of $I(S_2)$ corresponding to $\BP[\ell_2, I(S_2)]$, and $X_{\ell_1 + \ell_2 + 1}, \dots, X_{m}$ be the partition of $I(S_3)$ corresponding to $\BP[\ell_3, I(S_3)]$. We have $\Sigma(X_i) \leq t$ for every $i \in [m]$. Furthermore, since $S_1, S_2, S_3$ are disjoint, $X_1, \dots, X_m$ are also disjoint. Hence, $I$ can indeed be packed into $m$ bins.

    Conversely, assume that there exists a solution $X_1, \dots, X_m$ with cardinalities $b_1 \geq \dots \geq b_m$ such that $b_1 \leq (1-\eps)\cdot n / 2$ and $\sum_{i=1}^6 b_i < (1-\eps) \cdot n$.
    We later show the following claim.
    \begin{claim}\label{claim:ARC_buckets}
        There is a partition $B_1, B_2, B_3$ of $[m]$ such that  $\sum_{i \in B_k} b_i \leq (1-\eps)n/2$ for every $k \in \{1, 2, 3\}$.
    \end{claim}
    Assume that \cref{claim:ARC_buckets} holds and set $S_k \coloneq \bigcup_{i \in B_k} X_i$ and $\ell_k \coloneq |B_k|$ for $k \in \{1, 2, 3\}$.
    Then $S_1, S_2, S_3$ is a partition of $I$ such that every part has size at most $(1-\eps)n/2$.
    Let $\hat S_1, \hat S_2, \hat S_3$ be the index sets of $S_1, S_2, S_3$, and observe that $\hat S_1, \hat S_2, \hat S_3$ is a partition of $[n]$ such that every part has size at most $(1-\eps)n/2$.
    Thus, the families $\Ff_1, \Ff_2, \Ff_3$ corresponding to $\ell_1, \ell_2, \ell_3$ are a \textsc{Yes} instance of \threepartition, and the above algorithm outputs \textsc{Yes} for \binpacking. It remains to prove the above claim.

    \begin{claimproof}[Proof of \cref{claim:ARC_buckets}]
        We differentiate between two cases by comparing $b_1$ to $(1-2\eps)n/3$. In each case, we describe an iterative construction of $B_1, B_2, B_3$ with the invariant that
        \begin{equation}\label{eq:ARC_invariant}
            \sum_{i \in B_k} b_i \leq (1-\eps)n/2 \qquad \text{ for every } k \in \{1, 2, 3\}
        \end{equation}
        holds in each iterative step.

        \textbf{Case 1:} $b_1 \leq (1-2\eps)n/3$.
        Initialize $B_1 \gets \{1, 6\}$, $B_2 \gets \{2, 5\}$ and $B_3 \gets \{3, 4\}$ (see \cref{fig:b1_small}).
        Then, for every $j \in \{7, \dots, m\}$ assign $j$ to the set $B_k$ that minimizes $\sum_{i \in B_k} b_i$, i.e.~$B_k \gets B_k \cup \{j\}$ where $k = \argmin_{k \in \{1, 2, 3\}}\{\sum_{i \in B_k} b_i \}$.
        This way, we obtain a partition $B_1 \uplus B_2 \uplus B_3 = [m]$.
        We verify that \cref{eq:ARC_invariant} holds in every step of the construction.
        At initialization, we have
        \begin{align*}
            \sum_{i \in B_1} b_i & = b_1 + b_6 \leq b_1 + b_5,                                                                                                        \\
            \sum_{i \in B_2} b_i & = b_2 + b_5 \leq b_1 + b_5,                                                                                                        \\
            \sum_{i \in B_3} b_i & = b_3 + b_4 \leq \frac{1}{2}\sum_{i=1}^4 b_i \leq \frac{1}{2}\sum_{i=1}^6 b_i \leq (1-\eps)n/2 \tag{since $b_3+b_4 \leq b_1+b_2$}.
        \end{align*}
        Furthermore, since $b_1 \leq (1-2\eps)n/3$ (we are in Case 1), we can bound
        \begin{align*}
            b_1 + b_5
             & \leq b_1 + (n-b_1)/4 \tag{because $b_1 \geq \dots \geq b_5$} \\
             & \leq \frac{n}{4} + \frac{3}{4} (1-2\eps)\frac{n}{3}          \\
             & =(1-\eps)n/2.
        \end{align*}
        Hence, \cref{eq:ARC_invariant} holds at initialization.
        Now consider the step of the construction corresponding to assigning $j \in \{7, \dots, m\}$ to $B_k$. Before the assignment, we have $\sum_{k=1}^3 \sum_{i \in B_k} b_i = \sum_{i = 1}^{j-1} b_i \leq n$ and so, the set $B_k$ that minimizes $\sum_{i \in B_k} b_i$ has $\sum_{i \in B_k} b_i \leq n/3$. Furthermore, notice that $b_m \leq \dots \leq b_7 \leq n/7$. So, after assigning $j$ to $B_k$, we obtain $\sum_{i \in B_k} b_i \leq n/3 + n/7 \leq (1-\eps)n/2$  (since $\eps < \frac{1}{21}$), i.e.~\cref{eq:ARC_invariant} is indeed preserved throughout the steps.

        \textbf{Case 2:} $b_1 \geq (1-2\eps)n/3$.
        Initialize $B_1 \gets \{1\}$, $B_2 \gets \{2\}$ and $B_3 \gets \{3, 4\}$ (see \cref{fig:b1_large}).
        Then, for every $j \in \{5, \dots, m\}$ assign $j$ to the set $B_k$ that minimizes $\sum_{i \in B_k} b_i$, i.e.~$B_k \gets B_k \cup \{j\}$ where $k = \argmin_{k \in \{1, 2, 3\}}\{\sum_{i \in B_k} b_i \}$.
        This way, we obtain $B_1 \uplus B_2 \uplus B_3 = [m]$.
        We verify that \cref{eq:ARC_invariant} holds in every step of the construction.
        At initialization, we have
        \begin{align*}
            \sum_{i \in B_1} b_i & = b_1 \leq (1-\eps)n/2,                                                                                                            \\
            \sum_{i \in B_2} b_i & = b_2 \leq b_1 \leq (1-\eps)n/2,                                                                                                   \\
            \sum_{i \in B_3} b_i & = b_3 + b_4 \leq \frac{1}{2}\sum_{i=1}^4 b_i \leq \frac{1}{2}\sum_{i=1}^6 b_i \leq (1-\eps)n/2. \tag{since $b_3+b_4 \leq b_1+b_2$}
        \end{align*}
        So \cref{eq:ARC_invariant} indeed holds at initialization.
        Now consider the first step of the construction, i.e.~we assign $5$ to $B_k$ and obtain
        \begin{align*}
            \sum_{i \in B_k} b_i
             & = \min\{b_1, b_2, b_3 + b_4\} + b_5                                                                                                                                                \\
             & \leq \min\{b_2, b_3 + b_4\} + b_5                                                                                                                                                  \\
             & \leq \frac{b_2 +b_3 +b_4}{2} + b_5                                                                                                                                                 \\
             & =\frac{5}{8}\left( \left(\frac{4}{5} b_2 + \frac{1}{5} b_5\right) + \left(\frac{4}{5} b_3 + \frac{1}{5} b_5\right) + \left(\frac{4}{5} b_4 + \frac{1}{5} b_5\right) +  b_5 \right) \\
             & \leq \frac{5}{8} (b_2 + b_3 + b_4 +b_5) \tag{since $b_5 \leq b_4 \leq b_3 \leq b_2$}                                                                                               \\
             & \leq \frac{5}{8}(n-b_1)                                                                                                                                                            \\
             & \leq \frac{5}{8}(n-(1-2\eps)n/3) \tag{since we are in Case 2}                                                                                                                      \\
             & = (1+\eps)\frac{5}{12}n \leq (1-\eps)n/2
        \end{align*}
        where the last inequality holds because $\eps \leq \frac{1}{11}$. So \cref{eq:ARC_invariant} holds after the first iterative step.
        Now consider the step of the construction corresponding to assigning $j \in \{6, \dots, m\}$ to $B_k$. Before the assignment, we have $\sum_{k=1}^3 \sum_{i \in B_k} b_i = \sum_{i = 1}^{j-1} b_i \leq n$ and so, the set $B_k$ that minimizes $\sum_{i \in B_k} b_i$ has $\sum_{i \in B_k} b_i \leq n/3$. Furthermore, notice that
        \begin{align*}
            b_m \leq \dots \leq b_6 & \leq \frac{n-b_1}{5}                                   \\
                                    & \leq \frac{1}{5}(n - (1-2\eps)n/3) \tag{by assumption} \\
                                    & = (1+\eps)\frac{2}{15}n.
        \end{align*} So, after assigning $j$ to $B_k$, we obtain $\sum_{i \in B_k} b_i \leq n/3 + (1+\eps)\frac{2}{15}n \leq (1-\eps)n/2$  (since $\eps < \frac{1}{19}$), i.e.~\cref{eq:ARC_invariant} is indeed preserved throughout the steps.

        \begin{figure}
            \begin{minipage}[t]{0.43\textwidth}
                \centering
                \includegraphics[height=0.25\textheight]{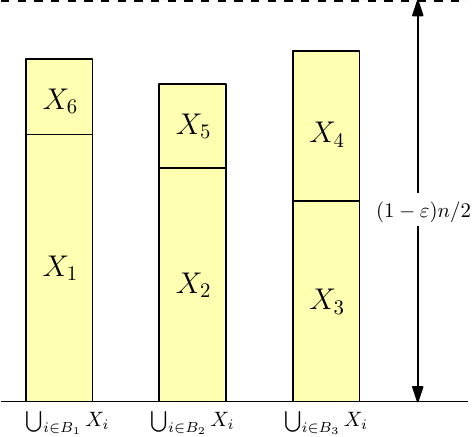}
                \caption{
                    In \cref{lem:case_ARC}, if $b_1 \le (1 - 2\eps)n/3$ then we initialize $B_1$, $B_2$ and $B_3$  to $B_1 \gets \{1, 6\}$, $B_2 \gets \{2, 5\}$ and $B_3 \gets \{3, 4\}$.} \label{fig:b1_small}
            \end{minipage} \hfill
            \begin{minipage}[t]{0.43\textwidth}
                \centering
                \includegraphics[height=0.25\textheight]{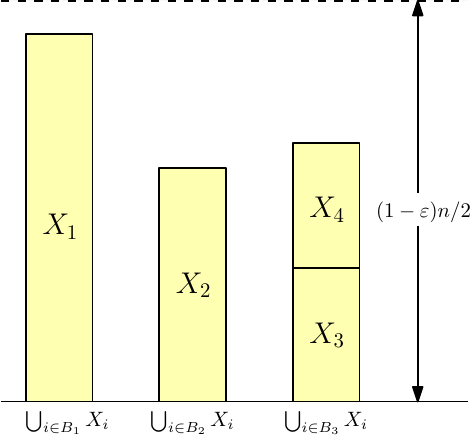}
                \caption{
                    In \cref{lem:case_ARC}, if $b_1 \ge (1 - 2\eps)n/3$ then we initialize $B_1$, $B_2$ and $B_3$ to $B_1 \gets \{1\}$, $B_2 \gets \{2\}$ and $B_3 \gets \{3, 4\}$.} \label{fig:b1_large}
            \end{minipage}
        \end{figure}
    \end{claimproof}
    This concludes the correctness proof.
\end{proof}
We now have all the ingredients to prove \cref{thm:bin-packing-arc}, that we restate for convenience.

\BPARC*
\begin{proof}
    Recall that we chose $\eps >0$ as a rational constant such that $\eps \leq \frac{1}{100}$ and $H(\eps) < \eps_6$, where $\eps_6$ is as in \cref{thm:nederlof_binpacking} and $H(\cdot)$ is the entropy function.
    Given any \binpacking instance $(m, t, I)$, we can solve it in the desired time using \cref{thm:nederlof_binpacking} if $m \leq 6$.
    This algorithm has no false positives, and correctly detects \text{Yes} instances with high probability.

    Otherwise, we have $m \geq 7$.
    If $|I| = n$ is odd, or $\eps n/2$ is not an integer, then we add dummy items of weight $0$ until the number of items is even and $\eps n/2$ is an integer, and let $(m,t,I')$ be the resulting instance.
    We have chosen $\eps$ as a rational, so $\eps = \frac{a}{b}$ for some $a,b \in \mathbb{N}$ such that $\gcd(a,b) = 1$.
    Let $n' \geq n$ be the smallest number such that $2b$ divides $n'$.
    Then, clearly $n'$ is even.
    Moreover, we have that $\eps n'/2 = \frac{a n'}{2b}$ is an integer.
    Hence, we must add at most $2b-1$ dummy items to $I$ to create instance $I'$.
    Since all added items have weight $0$, $(m,t,I)$ is a yes-instance if and only if $(m,t,I')$ is a yes-instance.

    If there is a solution of cardinalities $b_1 \geq \dots \geq b_m$, then one of the following cases holds.
    \begin{caseenum}
        \item $b_1 \geq (1+\eps)\cdot n / 2$
        \item $\sum_{i=1}^6 b_i \geq (1-\eps)\cdot n$
        \item $b_1 \in [\frac{n}{2} \pm \eps\cdot \frac{n}{2}]$ and $\sum_{i=1}^6 b_i < (1-\eps) \cdot n$
        \item $b_1 \leq (1-\eps)\cdot n / 2$ and $\sum_{i=1}^6 b_i < (1-\eps) \cdot n$
    \end{caseenum}
    Hence, since $m \geq 7$ and $n$ is even and $\eps n/2$ is an integer, we can run the algorithms of \cref{lem:case_ARC,lem:case_balanced,lem:case_cstBP,lem:case_DP}
    and output \textsc{No} if and only if all the algorithms output \textsc{No}.
    If there is a solution to the given \binpacking instance this algorithm outputs \textsc{Yes} with high probability, and \text{No} if there is no solution.

    Indeed, none of the four algorithms have false positives, so we always output \textsc{No} if there is no solution.
    If there is a solution, it falls into one of the four cases, and then we run an algorithm for the corresponding case, and this algorithm outputs \textsc{Yes} with high probability.
    Each algorithm used, as well as the algorithm of \cref{thm:nederlof_binpacking} runs in time $O((2-\gamma)^{|I'|})$, for some sufficiently small $\gamma > 0$, and hence the overall algorithm also runs in time $O((2-\gamma)^{|I'|}) = O((2-\gamma)^{|I| + 2b}) = O((2-\gamma)^{|I|})$, since $b$ is a constant, concluding the proof.
\end{proof}
 
\bibliography{references}

\end{document}